\documentclass[preprint,12pt]{elsarticle}

\usepackage{amssymb}
\usepackage{amsthm}
\usepackage{hyperref}
\hypersetup{
    colorlinks=true,
    linkcolor=blue,
    citecolor=blue,
    filecolor=blue,      
    urlcolor=blue,
    pdfpagemode=FullScreen,
    }
\usepackage{complexity}
\usepackage{float}
\usepackage{enumitem}%
\usepackage{graphicx}
\usepackage{tkz-graph}
\usepackage{mdframed}
\usepackage{tikz}
\usepackage{amsmath}

\usetikzlibrary{positioning}
\usetikzlibrary{fit,shapes.geometric}
\newcounter{nodemarkers}

\newtheorem{theorem}{Theorem}[section]
\newtheorem{definition}{Definition}[section]
\newtheorem{proposition}{Proposition}[section]
\newtheorem{lemma}{Lemma}[section] 
\newtheorem{corollary}[theorem]{Corollary}

\journal{Discrete Applied Mathematics}

\begin{document}
\makeatletter
\def\ps@pprintTitle{%
   \let\@oddhead\@empty
   \let\@evenhead\@empty
   \let\@oddfoot\@empty
   \let\@evenfoot\@empty}
\makeatother

\begin{frontmatter}


\author{Peace Ayegba\corref{label2}\fnref{label2}}
\ead{p.ayegba.1@research.gla.ac.uk}
\author{Sofiat Olaosebikan}
\ead{sofiat.olaosebikan@glasgow.ac.uk}
\author{David Manlove}
\ead{david.manlove@glasgow.ac.uk}
\cortext[label2]{Corresponding author}

\fntext[label2]{Supported by a College of Science and Engineering Scholarship, University of Glasgow, United Kingdom}

\affiliation{organization={School of Computing Science, University of Glasgow},
city={Glasgow},
country={United Kingdom}}

\title{Structural aspects of the Student Project Allocation problem}



\begin{abstract}
We study the Student Project Allocation problem with lecturer preferences over students ({\sc spa-s}), which involves the assignment of students to projects based on student preferences over projects, lecturer preferences over students, and capacity constraints on both projects and lecturers. The goal is to find a \emph{stable matching} that ensures no student and lecturer can mutually benefit by deviating from a given assignment to form an alternative arrangement involving some project. We explore the structural properties of {\sc spa-s} and characterise the set of stable matchings for an arbitrary {\sc spa-s} instance. We prove that, similar to the classical Stable Marriage problem ({\sc sm}) and the Hospital Residents problem ({\sc hr}), the set of all stable matchings in {\sc spa-s} forms a distributive lattice. In this lattice, the student-optimal and lecturer-optimal stable matchings represent the minimum and maximum elements, respectively. Our results extend known structural characterisations from bipartite models to the more complex {\sc spa-s} setting, and provide a basis for the development of efficient algorithms to address several open problems in {\sc spa-s} and its extensions.
\end{abstract}




\end{frontmatter}


\section{Introduction}
\label{sect:intro}
Matching problems have been extensively studied, particularly in the fields of theoretical computer science and economics, due to their broad range of applications. These problems typically arise when we need to allocate one set of agents to another, based on given preferences and subject to capacity constraints.
Given two disjoint sets \( A \) and \( B \), and a set of acceptable pairs \( E \subseteq A \times B \), a \textit{matching} is a subset \( M \subseteq E \) such that each element in \( A \) is matched to at most one element in \( B \), and the capacity constraints for elements in both sets are satisfied. Matching algorithms are widely used in various real-world applications, such as assigning kidney donors to patients \cite{DGGKMP2024,R2007,benedek2023}, junior residents to hospitals \cite{I1998, MOPU2007}, users to access points in wireless networks \cite{GSBDH2015}, and students to university projects \cite{CC2023, MMO2022, AIM2007}. For example, in the United Kingdom, the National Health Service (NHS) employs matching algorithms to pair living donors with compatible recipients, improving both the efficiency and success rates of kidney transplants \cite{NHSBT2024,NSF2024}. In two-sided matching problems, such as those involving students and lecturers, different types of optimality criteria can arise. Roth \cite{R1984} emphasized that a key criteria for a matching is \textit{stability}, which ensures that no two agents, who are not currently matched to each other, would both prefer being paired together over their current assignments. Such pairs are known as \textit{blocking pairs}.\\

\noindent In their seminal work \cite{GS1962}, Gale and Shapley introduced the Stable Marriage problem ({\sc sm}), which involves an equal number of men and women who each rank all members of the opposite sex in strict order of preference. The authors proved that every instance of {\sc sm} has at least one stable matching, and presented an algorithm that finds one in $O(n^2)$ time, where $n$ is the number of participants on each side. In the man-oriented version of the algorithm, men propose to women, who accept or reject proposals based on their preferences. This results in the \textit{man-optimal} stable matching, where each man is matched to his most preferred partner among all stable matchings, while each woman receives her least preferred partner. Conversely, when women propose to men, the algorithm yields the \textit{woman-optimal} stable matching. Relaxations of {\sc sm} include settings where preference lists are incomplete ({\sc smi}), preferences contains ties ({\sc smt}) or both incomplete lists and ties ({\sc smti}) \cite{GS1962,GI1989,IMMM2002}. \\

\noindent The Hospital Residents problem ({\sc hr}) is a many-to-one extension of {\sc sm} that models the assignment of junior residents to hospitals \cite{GS1962, MOPU2007, R1984}. In this setting,  each resident is assigned to at most one hospital, while each hospital can accept multiple residents up to its capacity. Using an extension of the Gale-Shapley algorithm, we can find both the resident-optimal and hospital-optimal matchings. Moreover, when preference lists include ties, the problem becomes the Hospitals Residents problem with Ties ({\sc hrt}) \cite{IMS2000,KM2014}. Another generalization of {\sc sm} is the \textit{Student-Project Allocation problem with Lecturer Preferences over Students} ({\sc spa-s}), which also generalizes {\sc hr}. This problem involves a set of students, projects, and lecturers, where each project is supervised by a unique lecturer, and both projects and lecturers are subject to capacity constraints. Students rank projects based on their preferences, while each lecturer ranks the students who find at least one of their offered projects acceptable.\\

\noindent Abraham et al. \cite{AIM2007} proved that every instance $I$ of {\sc spa-s} admits at least one stable matching and described two linear-time algorithms to find a stable matching in $I$. The first algorithm constructs the \textit{student-optimal} stable matching, where each student is assigned to the best project they could obtain in any stable matching. The second algorithm constructs the \textit{lecturer-optimal} stable matching, in which each lecturer is assigned a set of students that is at least as favourable (as precisely defined in Section \ref{sect:lect-pref}) as in any other stable matching. Additionally, the authors presented key properties that the set of stable matchings in any instance of {\sc spa-s} satisfies, analogous to the Rural Hospitals Theorem for {\sc hr} \cite{GI1989}, which they refer to as the \textit{Unpopular Projects Theorem} [see Theorem \ref{thm:unpopular-students}].  For a detailed example of the \textit{student-optimal} and \textit{lecturer-optimal} stable matchings, readers are referred to the example in Section \ref{sect:example1}.

\subsection{Related work}
\noindent Research has shown that stable matchings exhibit remarkably pristine structural properties, which have enabled the development of polynomial-time algorithms for computing stable matchings that are optimal with respect to specific criteria \cite{EIV2023,M2013,GMRV2023}. In general, there may be many stable matchings for a given instance of {\sc sm} \cite{IL1986}. Furthermore, the set of stable matchings forms a finite distributive lattice under the natural dominance relation \( \preceq \), a result attributed to John Conway by Knuth \cite{K1997}; similar results hold in {\sc hr} [\cite{GI1989}, Section 1.6.5]. A distributive lattice is a partially ordered set (poset) where the meet and join operations distribute over each other (see Definition \ref{def:dl} for a formal definition). \\

\noindent In the context of {\sc smti} and {\sc hrt}, three notions of stability are defined: \textit{weak stability}, \textit{strong stability}, and \textit{super-stability}. Manlove \cite{M2002} proved that the set of strongly stable matchings in {\sc smti} forms a finite distributive lattice. Later, Ghosal et al. \cite{GKP2015} presented a polynomial-time algorithm to enumerate all strongly stable matchings and provided an alternative proof demonstrating that the set of strongly stable matchings in {\sc smti} forms a distributive lattice. Similarly, Speiker \cite{S1995} showed that the set of super-stable matchings in {\sc smti} also forms a distributive lattice, with a subsequent proof later provided by Manlove \cite{M2002}. \\

\noindent The \textit{Workers Firms} problem ({\sc wf-2}) is a many-to-many extension of {\sc sm} and {\sc hr}, where each worker and firm can be matched with multiple partners. In this version, each worker ranks acceptable subsets of firms, and each firm ranks acceptable subsets of workers. Pairwise stability ensures that no unmatched worker-firm pair would both prefer to match with each other over their current matches. Blair \cite{B1988} showed that the set of pairwise-stable matchings forms a lattice under the assumption that preferences are path-independent. Later, Alkhan \cite{A2002} and Li \cite{L2014} extended this result by proving that, under substitutable, strict, and cardinally monotone preferences, the set of pairwise-stable matchings forms a distributive lattice.\\

\noindent In each of these results, leveraging the lattice structure facilitated the development of efficient algorithms and \NP-hardness results for other stable matching problems within that setting. For example, in {\sc sm}, the lattice structure led to algorithms for finding all stable pairs \cite{G1987}, generating all stable matchings \cite{G1987}, finding egalitarian or minimum regret stable matchings \cite{G1987}, and establishing the \P-completeness of the problem of counting stable matchings \cite{IL1986}.\\

\noindent  Another generalization of {\sc hr} is  the Laminar Classified Stable Matching ({\sc lcsm}) model introduced by Huang~\cite{huang2010classified}, in which institutes and applicants have preferences over one another. Each institute classifies applicants into a laminar family of classes and specifies upper and lower bounds on the number of applicants it can accept from each class.  Its extension, {\sc 2lcsm} \cite{fleiner2016lower}, adds quota constraints on both sides, by assigning each agent a separate matroid that defines their feasible assignments; the overall feasibility condition is then given by the direct sum of these individual matroids. As noted by the authors, {\sc lcsm} reduces to a special case of {\sc spa-s} when classifications are simple partitions and no lower bounds are imposed.  Importantly, neither {\sc lcsm} nor {\sc 2lcsm} can model the case where a student may be assigned to different projects offered by the same lecturer across stable matchings. \\

\noindent The classical stable marriage problem has also been generalised through matroid-theoretic frameworks. One such model is Fleiner’s formulation of bipartite stable matchings \cite{fleiner2001matroid, fleiner2003fixed}, which characterises stable matchings as matroid kernels. The {\sc spa-s} model can be embedded into this framework by representing students as a partition matroid and lecturers as a truncation of a direct sum of uniform matroids, as observed in \cite{AIM2007}. In this representation, the bipartite graph becomes a multigraph: vertices on one side correspond to students, those on the other to lecturers, and edges to acceptable student–project pairs.

\bigskip
\noindent Yu further generalises Fleiner’s model by introducing the framework of $g$-matroids and $g$-polymatroids~\cite{yokoi2015stable}, which extend the standard matroid setting to accommodate more flexible assignment constraints, including lower quotas. In contrast to classical matroids, $g$-matroids do not require the empty set to be feasible, thereby allowing the model to capture settings in which each agent must be assigned a minimum number of partners in order for the overall matching to be considered feasible.

\medskip
\noindent In both Fleiner’s and Yu’s models, a structural result holds that generalises the Rural Hospitals Theorem: if the matroid (or g-matroid) kernel corresponding to one stable matching does not span the entire ground set, then the same subset is spanned in all stable matchings. This implies that if an agent is undersubscribed in one stable matching, they must be assigned the same set of agents in every stable matching. In contrast, this property does not hold in {\sc spa-s}, where a project or lecturer may be undersubscribed in one stable matching but assigned different sets of students in others. This suggests that although {\sc spa-s} may be embedded into these frameworks, key structural properties do not carry over.

\subsection{Our contributions}
\noindent The main contribution of this paper is to prove that the set of stable matchings in \textsc{spa-s} forms a distributive lattice, with the student-optimal and lecturer-optimal matchings representing the minimum and maximum elements of the lattice, respectively. This result builds upon the preliminary work of Olaosebikan \cite{O2020}, which explored the structure of stable matchings under the restriction that each student ranks projects offered by different lecturers. In this paper, we confirm that this characterisation of the set of stable matchings in \textsc{spa-s} remains valid without this restriction, by establishing the result in the general setting. \\

\noindent To achieve this, we introduce three new results—namely Lemmas \ref{lem:samelecturer1}, \ref{lem:student-pref-lecturer-pref}, and \ref{lem:s-prefers-l-prefers1} —that establish interesting properties required for our final result. Moreover, we provide revised proofs for the meet and join operations, addressing gaps in Olaosebikan \cite{O2020} and ensuring correctness in the general case. Overall, our contributions offer new insights that enhance our understanding of the structure of stable matchings in \textsc{spa-s}.

\subsection{Structure of the paper}
\noindent The paper is structured as follows. In Section \ref{sect:premdef}, we provide formal definitions relevant to this paper. In Section~\ref{sect:structural-lemmas}, we present some structural properties of stable matchings in {\sc spa-s}. Section~\ref{sect:stablematchings-distlattice} contains our main result, showing that the set of stable matchings forms a distributive lattice. We conclude in Section~\ref{sect:conclusion} with a summary of our results and a discussion of open problems.

\section{Preliminary Definitions}
\label{sect:premdef}
\noindent Section \ref{def:spa-s} provides the formal definition of the \textsc{spa-s} model, as introduced in \cite{AIM2007}. In Section \ref{def:pom}, we present the concept of student and lecturer preferences over matchings, also from \cite{AIM2007}. Finally, in Section \ref{def:dr}, we define the dominance relation on stable matchings and prove that it is reflexive, anti-symmetric, and transitive, as presented in \cite{O2020}.

\subsection{The Student-Project Allocation model}
\label{def:spa-s}
\noindent An instance $I$ of \textsc{spa-s} involves a set $\mathcal{S} = \{s_1, s_2, \ldots, s_{n_1}\}$ of students, a set $\mathcal{P} = \{p_1, p_2, \ldots, p_{n_2}\}$ of projects, and a set $\mathcal{L} = \{l_1, l_2, \ldots, l_{n_3}\}$ of lecturers. Each student $s_i \in \mathcal{S}$ ranks a subset of $\mathcal{P}$ in strict order, which forms $s_i$'s \textit{preference list}. A project $p_j \in \mathcal{P}$ is said to be \textit{acceptable} to $s_i$ if it appears in $s_i$'s preference list. The set of all projects acceptable to $s_i$ is denoted by $A_i \subseteq \mathcal{P}$. \\

\noindent Each project $p_j \in \mathcal{P}$ is offered by a single lecturer $l_k \in \mathcal{L}$. Each lecturer $l_k$ offers a non-empty subset of projects $P_k \subseteq \mathcal{P}$, where the sets $P_1, P_2, \ldots, P_{n_3}$ form a partition of $\mathcal{P}$. Each lecturer $l_k \in \mathcal{L}$ ranks, in strict order of preference, those students who find at least one project in $P_k$ acceptable. A student $s_i$ is said to be \emph{acceptable} to $l_k$ if $s_i$ appears in $l_k$'s preference list. The set of all students acceptable to $l_k$ is denoted by $\mathcal{L}_k \subseteq \mathcal{S}$. Let $p_j, p_{j'} \in \mathcal{P}$ be projects such that $s_i$ finds both acceptable, i.e., $p_j, p_{j'} \in A_i$. If $p_j$ precedes $p_{j'}$ in $s_i$'s preference list, then we say that $s_i$ \emph{prefers} $p_j$ to $p_{j'}$. Similarly, if $s_i, s_{i'} \in \mathcal{S}$ are acceptable to $l_k$, and $s_i$ precedes $s_{i'}$ in $l_k$'s preference list, then we say that $l_k$ \emph{prefers} $s_i$ to $s_{i'}$. For any pair $(s_i, p_j) \in \mathcal{S} \times \mathcal{P}$ where $p_j$ is offered by $l_k$, we define $(s_i, p_j)$ to be an \emph{acceptable pair} if $p_j \in A_i$ and $s_i \in \mathcal{L}_k$, i.e., both $s_i$ finds $p_j$ acceptable and $l_k$ finds $s_i$ acceptable. \\

\noindent Further, each project $p_j \in \mathcal{P}$ has a capacity $c_j \in \mathbb{Z^+}$, representing the maximum number of students that can be assigned to $p_j$. Similarly, each lecturer $l_k$ has a capacity $d_k \in \mathbb{Z^+}$, indicating the maximum number of students who may be assigned to projects offered by $l_k$. For any lecturer $l_k$, it holds that $\max \{c_j : p_j \in P_k\} \leq d_k \leq \sum \{c_j : p_j \in P_k\}$. This implies that the capacity of $l_k$ is (i) at least the maximum capacity of the projects offered by $l_k$, and (ii) at most the sum of the capacities of all projects that $l_k$ offers. We denote by $\mathcal{L}_k^j$ the \emph{projected preference list} of $l_k$ for $p_j$, which is derived from $\mathcal{L}_k$ by removing those students who do not find $p_j$ acceptable. \\

\noindent An assignment $M$ in $I$ is a subset of $\mathcal{S} \times \mathcal{P}$ such that, for each pair $(s_i, p_j) \in M$, $s_i$ finds $p_j$ acceptable. The \emph{size} of $M$ is the number of (student, project) pairs in $M$, denoted by $|M|$. We say that $s_i$ is \emph{assigned} to $p_j$ and that $p_j$ is assigned to $s_i$ if $(s_i, p_j) \in M$. We denote by $M(s_i)$ the set of projects assigned to a student $s_i \in \mathcal{S}$, and by $M(p_j)$ the set of students assigned to a project $p_j \in \mathcal{P}$. For simplicity, if $s_i$ is assigned to a project $p_j$ offered by a lecturer $l_k$, we say that $s_i$ is \emph{assigned to} $l_k$, and $l_k$ is \emph{assigned to} $s_i$. Consequently, we denote by $M(l_k)$ the set of students assigned to $l_k$ in $M$. Furthermore, we denote by $S_k(M,M')$ the set of students assigned in two different stable matchings $M$ and $M'$ to different projects offered by the same lecturer $l_k$. A project $p_j \in \mathcal{P}$ is said to be \emph{undersubscribed}, \emph{full}, or \emph{oversubscribed} in $M$ if $|M(p_j)| < c_j$, $|M(p_j)| = c_j$, or $|M(p_j)| > c_j$, respectively. Similarly, a lecturer $l_k \in \mathcal{L}$ is said to be \emph{undersubscribed}, \emph{full}, or \emph{oversubscribed} in $M$ if $|M(l_k)| < d_k$, $|M(l_k)| = d_k$, or $|M(l_k)| > d_k$, respectively. \\

\subsubsection{Example.}
\label{sect:example1}
\noindent An example {\sc spa-s} instance is shown in Figure \ref{fig:example1}. Here, the set of students is $\mathcal{S} = \{s_1, s_2, \ldots, s_5\}$, the set of projects is $\mathcal{P} = \{p_1, p_2, \ldots, p_5\}$, and the set of lecturers is $\mathcal{L} = \{l_1, l_2\}$. Each student has a preference list over the projects they find acceptable. For example, $s_1$'s preference list is $p_1, p_2$, and $s_2$'s preference list is $p_2, p_3$. Also, lecturer $l_1$ offers $p_1, p_2, p_5$, while lecturer $l_2$ offers $p_3, p_4$. Each lecturer ranks students in order of preference. In this example, $l_1$'s preference list is $s_4, s_5, s_3, s_1, s_2$, and the projected preference list of $l_1$ for $p_1$ includes $s_3, s_1$, ranked in that order.

\begin{figure}[H]
\centering
\small
\renewcommand{\arraystretch}{1.2} 
\setlength{\tabcolsep}{4pt} 
\begin{tabular}{p{0.35\textwidth} p{0.35\textwidth} p{0.25\textwidth}}
\hline
\textbf{Students' preferences} & \textbf{Lecturers' preferences} & \textbf{Offers} \\ 
\hline
$s_1$: $p_1 \; p_2$ & $l_1$: $s_4 \; s_5 \; s_3 \; s_1 \; s_2$ & $p_1$, $p_2$, $p_5$ \\ 
$s_2$: $p_2 \; p_3$ & $l_2$: $s_2 \; s_3 \; s_5 \; s_4$ & $p_3$, $p_4$ \\ 
$s_3$: $p_3 \; p_1$ & & \\ 
$s_4$: $p_4 \; p_5$ & & \\ 
$s_5$: $p_5 \; p_4$ & & \\ 
\multicolumn{2}{l}{\textbf{Project capacities:} $c_1 = c_2 = c_3 = c_4 = c_5 = 1$} & \\
\multicolumn{2}{l}{\textbf{Lecturer capacities:} $d_1 = 3, \; d_2 = 2$} & \\ 
\hline
\end{tabular}
\caption{An instance $I_1$ of {\sc spa-s}}
\label{fig:example1}
\end{figure}

\noindent \textbf{Matching:}  
A matching \( M \) is an assignment of students to projects such that the following conditions hold:  
\begin{enumerate}
    \item[(a)] Each student \( s_i \in \mathcal{S} \) is assigned to at most one project, i.e., \( |M(s_i)| \leq 1 \).
    \item[(b)] No project \( p_j \in \mathcal{P} \) exceeds its capacity, i.e., \( |M(p_j)| \leq c_j \).
    \item[(c)] No lecturer \( l_k \in \mathcal{L} \) exceeds their capacity, i.e., \( |M(l_k)| \leq d_k \).
\end{enumerate}

\noindent \textbf{Blocking Pairs and Stable Matchings.}  
Let \( I \) be an instance of {\sc spa-s} and let \( M \) be a matching in \( I \). An acceptable pair \( (s_i, p_j) \in (\mathcal{S} \times \mathcal{P}) \setminus M \) is a \emph{blocking pair} for \( M \) if either:
\begin{itemize}
    \item[(S1)] \( s_i \) is unassigned in \( M \), or
    \item[(S2)] \( s_i \) prefers \( p_j \) to \( M(s_i) \),
\end{itemize}
\noindent and one of the following conditions holds:
\begin{itemize}
    \item[(P1)] \( p_j \) is undersubscribed in \( M \), and its lecturer \( l_k \) is also undersubscribed;
    \item[(P2)] \( p_j \) is undersubscribed in \( M \), \( l_k \) is full in $M$, and \( s_i \in M(l_k) \);  (Note that this cannot arise when (S1) holds.)
    \item[(P3)] \( p_j \) is undersubscribed in \( M \), \( l_k \) is full in $M$, and \( l_k \) prefers \( s_i \) to the worst student in \( M(l_k) \);
    \item[(P4)] \( p_j \) is full in \( M \), and \( l_k \) prefers \( s_i \) to the worst student in \( M(p_j) \).
\end{itemize}

\noindent If such a pair exists, we say that \( (s_i, p_j) \) \emph{blocks} \( M \). A matching \( M \) is \emph{stable} if it admits no blocking pair.\\

\noindent With respect to the {\sc spa-s} instance \( I_1 \) shown in Figure \ref{fig:example1}, the matching  
\( M_1 = \{(s_1, p_1), (s_2, p_2), (s_3, p_3), (s_4, p_4), (s_5, p_5)\} \) is a stable matching, as it does not admit any blocking pair.

\subsection{Preferences over Matchings}
\label{def:pom}

\noindent In this section, we extend the notion of preferences over individual assignments for students and lecturers to preferences over matchings. First, we present the Unpopular Projects Theorem, which captures key structural properties of the set of stable matchings in {\sc spa-s}.

\begin{theorem}[Unpopular Projects Theorem {\cite{AIM2007, O2020}}]
\label{thm:unpopular-students}
Let \( \mathcal{M} \) denote the set of all stable matchings in a given instance of {\sc spa-s}. Then:
\begin{itemize}
\item[(i)] Each lecturer is assigned the same number of students in all stable matchings in \( \mathcal{M} \).
\item[(ii)] Exactly the same students are unassigned in all stable matchings in \( \mathcal{M} \).
\item[(iii)] Any project offered by an undersubscribed lecturer is assigned the same number of students in all stable matching in \( \mathcal{M} \).
\end{itemize}
\end{theorem}

\noindent In the Rural Hospitals Theorem for the {\sc hr} model, an undersubscribed hospital is assigned the same set of residents in every stable matching, and each hospital receives the same number of residents across all stable matchings. However, these properties do not fully extend to {\sc spa-s}. In particular:
\begin{itemize}
    \item An undersubscribed lecturer may be assigned different sets of students in different stable matchings (see Figure~3 in \cite{AIM2007}).
    \item A project offered by a full lecturer in one stable matching may be assigned a different number of students in another stable matching (see Figure~4 in \cite{AIM2007}).
\end{itemize}
\noindent Nevertheless, the properties in Theorem \ref{thm:unpopular-students} hold across all stable matchings in any instance of {\sc spa-s}.

\subsubsection{Student Preferences over Matchings}
\noindent Let \( I \) be an instance of {\sc spa-s}, and let \( \mathcal{M} \) denote the set of all stable matchings in \( I \). Given two matchings \( M, M' \in \mathcal{M} \), a student \( s_i \in \mathcal{S} \) \emph{prefers} \( M \) to \( M' \) if \( s_i \) is assigned in both matchings and prefers \( M(s_i) \) to \( M'(s_i) \). Similarly, \( s_i \) is \emph{indifferent} between \( M \) and \( M' \) if either:
\begin{enumerate}[label=(\roman*)]
    \item \( s_i \) is unassigned in both \( M \) and \( M' \), or
    \item \( s_i \) is assigned the same project in both matchings, i.e., \( M(s_i) = M'(s_i) \).
\end{enumerate}

\subsubsection{Lecturer Preferences over Matchings}
\label{sect:lect-pref}
\noindent It is not immediately clear how to compare two stable matchings from the perspective of a lecturer. To formalise lecturer preferences over matchings, we adopt the definition proposed by Abraham et al. in ~\cite{AIM2007}. Let \( M \) and \( M' \) be two stable matchings in \( \mathcal{M} \). By Theorem~\ref{thm:unpopular-students}, \( |M| = |M'| \) and \( |M(l_k)| = |M'(l_k)| \) for each lecturer \( l_k \). Suppose that \( l_k \) is assigned different sets of students in \( M \) and \( M' \). Define
\[
M(l_k) \setminus M'(l_k) = \{s_1, \ldots, s_r\}, \quad 
M'(l_k) \setminus M(l_k) = \{s'_1, \ldots, s'_r\},
\]
where the students in each set are listed in the order they appear in \( l_k \)'s preference list \( \mathcal{L}_k \). Then \( l_k \) \emph{prefers} \( M \) to \( M' \) if \( l_k \) prefers \( s_i \) to \( s'_i \) for all \( i \in \{1, \ldots, r\} \). On the other hand, lecturer $l_k$ is \emph{indifferent} between $M$ and $M'$ if $l_k$ is not assigned to any student or is assigned the same set of students in $M$ and $M'$, i.e., $M(l_k) = M'(l_k)$.

\medskip
\noindent\textbf{Example.} Consider the two stable matchings \( M_1 \) and \( M_2 \) for instance \( I_1 \). Then:
\[
M_2(l_1) \setminus M_1(l_1) = \{s_4, s_3\}, \quad M_1(l_1) \setminus M_2(l_1) = \{s_5, s_2\}.
\]
The reader can verify that neither \( l_1 \) nor \( l_2 \) is assigned their most preferred set of students in both stable matchings. However, since \( l_1 \) prefers \( s_4 \) to \( s_5 \) and \( s_3 \) to \( s_2 \), it follows that \( l_1 \) prefers \( M_2 \) to \( M_1 \).

\subsection{Dominance relation}
\label{def:dr}
\noindent We now define the dominance relation that plays a central role in constructing the lattice structure of stable matchings. Let \( \mathcal{M} \) denote the set of all stable matchings in {\sc spa-s}. We show in Proposition~\ref{prop:dominance} that \( \mathcal{M} \), under the dominance relation \( \preceq \), forms a partial order. Unless stated otherwise, whenever we write \( M \preceq M' \), we refer to the \emph{student-oriented} dominance relation. References to the lecturer-oriented dominance relation will be made explicit.

\begin{definition}[\textbf{Student-oriented dominance relation}]
Let \( M, M' \in \mathcal{M} \). We say that \( M \) \emph{dominates} \( M' \), denoted \( M \preceq M' \), if and only if each student prefers \( M \) to \( M' \), or is indifferent between them.
\end{definition}

\noindent\textbf{Example.} Consider instance \( I_1 \) in Figure~\ref{fig:example1}, which admits the following two stable matchings: \( M_1 = \{(s_1,p_1), (s_2,p_2), (s_3,p_3), (s_4,p_4), (s_5,p_5)\}, 
M_2 = \{(s_1,p_2), (s_2,p_3), (s_3,p_1), (s_4,p_5), (s_5,p_4)\}. \)
Each student prefers their assignment in \( M_1 \) to their assignment in \( M_2 \), so \( M_1 \) dominates \( M_2 \). 

\begin{definition}[\textbf{Lecturer-oriented dominance}]
Let \( M, M' \in \mathcal{M} \). We say that \( M \) dominates \( M' \) from the lecturers’ perspective if each lecturer either prefers \( M \) to \( M' \), or is indifferent between the two.
\end{definition}

\noindent We note that in the hospital-resident setting, given any two stable matchings \( M \) and \( M' \), each hospital either prefers all of its assigned residents in \( M \) to those assigned to it in \( M' \setminus M \), or prefers all its assigned residents in \( M' \) to those assigned to it in \( M \setminus M' \). This property does not hold in {\sc spa-s}. In {\sc spa-s}, if some lecturer \( l \) is assigned different sets of students in two stable matchings \( M \) and \( M' \), they may not prefer all students in \( M(l) \) to those in \( M'(l) \setminus M(l) \), nor all students in \( M'(l) \) to those in \( M(l) \setminus M'(l) \). However, it is always the case that \( l \) prefers at least one student in \( M(l) \setminus M'(l) \) to at least one student in \( M'(l) \setminus M(l) \), or vice versa. 

\medskip
\noindent In Figure \ref{fig:example1}, \( M_1 \) is the \textit{student-optimal} stable matching since every student is assigned to their best project in \( M_1 \). Similarly, the matching  
\( M_2 = \{(s_1, p_2), (s_2, p_3), (s_3, p_1), (s_4, p_5), (s_5, p_4)\} \) is also stable in \( I_1 \), and \( M_2 \) is the \textit{lecturer-optimal} stable matching. Clearly, in \( M_2 \), each lecturer is assigned a student they prefer to at least one of the students assigned to them in \( M_1 \). 

\medskip
\noindent In Lemma~\ref{lem:samelecturer1}, we prove that for any two stable matchings $M$ and $M'$, if a student is assigned to lecturer $l_k$ in both matchings, then there exists at least one student in $M'(l_k) \setminus M(l_k)$ and, consequently, one in $M(l_k) \setminus M'(l_k)$. We then prove in in Lemma~\ref{lem:student-pref-lecturer-pref} that if some student $s \in M(l_k) \setminus M'(l_k)$ prefers $M$ to $M'$, then $l_k$ prefers $M'$ to $M$.

\begin{proposition}
\label{prop:dominance}
Let \( \mathcal{M} \) be the set of all stable matchings in \( I \). The dominance relation \( \preceq \) defines a partial order on \( \mathcal{M} \), and we denote this partially ordered set as \( (\mathcal{M}, \preceq) \).
\end{proposition}

\noindent We remark that the proof given below follows a similar line of argument to that presented in \cite{O2020}.

\begin{proof}
We show that the dominance relation \( \preceq \) on \( \mathcal{M} \) is:  
(i) reflexive, (ii) anti-symmetric, and (iii) transitive.  

\begin{enumerate}[label=(\roman*)]
    \item \textbf{Reflexive:}  
    Let \( M \in \mathcal{M} \). Clearly, \( M \preceq M \), since every student is indifferent between \( M \) and itself. Thus, \( \preceq \) is reflexive.

    \item \textbf{Anti-symmetric:}  
    Let \( M, M' \in \mathcal{M} \) such that \( M \preceq M' \) and \( M' \preceq M \). Then \( M = M' \). Suppose, for contradiction, that \( M \neq M' \). Then there exists some student \( s_i \) such that \( s_i \) is assigned in both \( M \) and \( M' \), and \( M(s_i) \neq M'(s_i) \). Since \( M \preceq M' \), \( s_i \) prefers \( M(s_i) \) to \( M'(s_i) \). Similarly, \( M' \preceq M \) implies \( s_i \) prefers \( M'(s_i) \) to \( M(s_i) \). This is a contradiction. Hence, \( M = M' \), and \( \preceq \) on $\mathcal{M}$ is anti-symmetric.  

    \item \textbf{Transitive:}  
    Let \( M, M', M'' \in \mathcal{M} \) such that \( M \preceq M' \) and \( M' \preceq M'' \). We claim that \( M \preceq M'' \). By Theorem \ref{thm:unpopular-students}, we know that exactly the same students are unassigned in all stable matchings. Thus, every student who is unassigned in $M$ is unassigned in $M''$, and every unassigned student is indifferent between $M$ and $M''$. Clearly, every student who is assigned to the same project in $M$ and $M''$ is indifferent between $M$ and $M''$.

    \medskip
    \noindent Now, let $s_i$ be some student who is assigned to different projects in both $M$ and $M''$, say $M(s_i)$ and $M''(s_i)$ respectively. First, suppose that $M(s_i) \neq M'(s_i)$; since $M \preceq M'$, it follows that $s_i$ prefers $M(s_i)$ to $M'(s_i)$. Further, (a) if $M'(s_i) = M''(s_i)$ then $s_i$ prefers $M(s_i)$ to $M''(s_i)$, and (b) if $M'(s_i) \neq M''(s_i)$, $M' \preceq M''$ implies that $s_i$ prefers $M'(s_i)$ to $M''(s_i)$, and since the preference lists are strictly ordered, $s_i$ prefers $M(s_i)$ to $M''(s_i)$. Now, suppose that $M'(s_i) = M''(s_i)$. It follows that $M'(s_i) \neq M''(s_i)$; thus $M' \preceq M''$ implies that $s_i$ prefers $M'(s_i)$ to $M''(s_i)$. This implies that $s_i$ prefers $M(s_i)$ to $M''(s_i)$. Hence our claim holds; and therefore $\preceq$ on $\mathcal{M}$ is transitive.
\end{enumerate}
\end{proof}

\begin{definition}[Distributive lattice \cite{GI1989}]
\phantomsection
\label{def:dl}
\noindent Let $A$ be a set and let $\preceq$ be an ordering relation defined on $A$. The partial order $(A,\preceq)$ is a distributive lattice if:
    \begin{enumerate}[label = (\roman*)]
        \item each pair of element $x,y \in A$ has a greatest lower bound, or meet, denoted $x \land y$, such that $x \land y \preceq x$, $x \land y \preceq y$, and there is no element $z \in A$ for which $z \preceq x$, $z \preceq y$ and $x \land y \preceq z$;
        
        \item each pair of element $x,y \in A$ has a least upper bound, or join, denoted $x \lor y$, such that $x \preceq x \lor y$, $y \preceq x \lor y$, and there is no element $z \in A$ for which $x \preceq z$, $y \preceq z$ and $z \preceq x \lor y$;
        \item the join and meet distribute over each other, i.e., for $x,y,z \in A$, $x \lor (y \land z) = (x \lor y) \land (x \lor z)$ and $x \land (y \lor z) = (x \land y) \lor (x \land z)$.
    \end{enumerate}

\end{definition}

\section{Structural Properties of Stable Matchings}
\label{sect:structural-lemmas}
\noindent In this section, we present new results that illustrate lecturers' preferences over matchings when a student prefers one stable matching to another. These results will be used in the next section to prove that the set of stable matchings forms a distributive lattice. We first present Proposition~\ref{prop:s-pref-l-prefs}, which is used in the proofs of Lemmas~\ref{lem:samelecturer1} and~\ref{lem:student-pref-lecturer-pref}. 

\medskip
\noindent
Let $M$ and $M'$ be two stable matchings in a {\sc spa-s} instance $I$.  
In Lemma~\ref{lem:samelecturer1}, we show that if a student \( s_i \) is assigned to different projects offered by the same lecturer \( l_k \) in \( M \) and \( M' \), and \( s_i \) prefers \( M \) to \( M' \), then \( l_k \) prefers some student in \( M'(l_k) \setminus M(l_k) \) to \( s_i \).  
In Lemma~\ref{lem:student-pref-lecturer-pref}, we show that if there exists a student \( s \in M(l_k) \setminus M'(l_k) \) who prefers \( M \) to \( M' \), then \( l_k \) prefers \( M' \) to \( M \).  
In Lemma~\ref{lem:s-prefers-l-prefers1}, we show that if \( s_i \) is assigned to \( p_j \) offered by \( l_k \) in \( M' \) and prefers \( M \) to \( M' \), then \( l_k \) prefers \( s_i \) to each student in \( M(p_j) \setminus M'(p_j) \), or, if \( p_j \) is undersubscribed in \( M \), to each student in \( M(l_k) \setminus M'(l_k) \).  
Finally, in Lemma~\ref{lem:samelecturer2}, we prove the symmetric case of Lemma~\ref{lem:samelecturer1}: if \( s_i \) is assigned to different projects offered by the same lecturer \( l_k \), then \( l_k \) prefers \( s_i \) to some student in \( M(l_k) \setminus M'(l_k) \).

\begin{proposition}
\label{prop:s-pref-l-prefs}
Let $M$ and $M'$ be two stable matchings in $I$, and let $s$ be some student assigned in $M$ to a project $p_j$ offered by lecturer $l_k$. If $s$ prefers $M$ to $M'$ and either $s \in M'(l_k)$ or $l_k$ prefers $s$ to at least one student in $M'(l_k)$, then $p_j$ is full in $M'$.
\end{proposition}

\begin{proof}
Let $s$ be some student assigned in $M$ to $p_j$ offered by $l_k$, where $s$ prefers $M$ to $M'$. Suppose, for contradiction, that $p_j$ is undersubscribed in $M'$. Then, if $s \in M'(l_k)$ or $l_k$ prefers $s$ to some student in $M'(l_k)$, it follows that $(s,p_j)$ forms a blocking pair in $M'$. This contradicts the stability of $M'$. Hence, $p_j$ is full in $M'$ and our claim holds.
\end{proof}

\begin{lemma}
\label{lem:samelecturer1}
Let \( M \) and \( M' \) be two stable matchings in \( I \). If some student \( s_i \) is assigned in \( M \) and \( M' \) to different projects offered by the same lecturer \( l_k \), and $s_i$ prefers $M$ to $M'$, then there exists some other student $s_r \in M'(l_k) \setminus M(l_k)$ such that $l_k$ prefers $s_r$ to $s_i$. Thus, \( M(l_k) \neq M'(l_k) \).
\end{lemma}

\begin{proof}
Let $M$ and $M'$ be two stable matchings in $I$. Let \( s_i \) be some student assigned to different projects in \( M \) and \( M' \), both offered by the same lecturer \( l_k \), and suppose \( s_i \) prefers \( M \) to \( M' \). Suppose for contradiction that \( M(l_k) = M'(l_k) \), i.e., \( M'(l_k) \setminus M(l_k) = \emptyset \). We construct sequences \( \langle s_1, s_2, s_3, \ldots \rangle \) and \( \langle p_0, p_1, p_2, \ldots \rangle \) of students and projects such that for each \( t \geq 2 \):
\begin{enumerate}[label=(\arabic*)]
    \item \( s_t \) prefers \( p_t \) to \( p_{t-1} \),
    \item \( (s_t, p_t) \in M \setminus M' \) and \( (s_t, p_{t-1}) \in M' \setminus M \),
    \item \( l_k \) offers both \( p_t \) and \( p_{t-1} \),
    \item \( l_k \) prefers \( s_t \) to \( s_{t-1} \).
\end{enumerate}

\noindent We prove by induction that these properties hold for all \( t \geq 2 \). Let $s_1 = s_i$.

\medskip
\noindent \textbf{Base case (\( t = 2 \)).}  
Let \( p_1 = M(s_1) \), \( p_0 = M'(s_1) \). Since \( s_1 \) prefers \( M \) to \( M' \), it follows that \( s_1 \) prefers \( p_1 \) to \( p_0 \). Moreover, \( (s_1, p_1) \in M \setminus M' \), \( (s_1, p_0) \in M' \setminus M \), and both projects are offered by \( l_k \). Since \( M' \) is stable, one of the following conditions hold:
\begin{enumerate}[label=(\roman*)]
    \item \( p_1 \) is full in \( M' \), and \( l_k \) prefers the worst student in \( M'(p_1) \) to \( s_1 \); or
    \item \( p_1 \) is undersubscribed in \( M' \), \( l_k \) is full in \( M' \), \( s_1 \notin M'(l_k) \), and \( l_k \) prefers the worst student in \( M'(l_k) \) to \( s_1 \).
\end{enumerate}

\noindent Since \( s_1 \in M'(l_k) \), it follows from Proposition~\ref{prop:s-pref-l-prefs} that \( p_1 \) is full in \( M' \) and case~(i) holds. In this case, since \( (s_1, p_1) \in M \setminus M' \), there exists some student \( s_2 \in M'(p_1) \setminus M(p_1) \); otherwise, \( p_1 \) would be oversubscribed in \( M \). Furthermore, \( l_k \) prefers \( s_2 \) to \( s_1 \). Now, in \( M \), \( s_2 \) must be assigned to some project \( p_2\) such that $s_2$ prefers $p_2$ to $p_1$; otherwise, \( (s_2, p_1) \) would block \( M \). Hence \( (s_2, p_2) \in M \setminus M' \). Moreover, since \( s_2 \in M'(l_k) \), and we assumed \( M(l_k) = M'(l_k) \), it follows that \( l_k \) also offers \( p_2 \). Thus, properties (1)–(4) hold for \( t = 2 \), completing the base case.

\medskip
\noindent \textbf{Inductive step.}  
Assume that properties (1) - (4) above hold for some $t=q-1 \geq 2$. We now show that the properties also hold for \( t = q \). By the inductive hypothesis:
\begin{enumerate}[label=(\arabic*)]
    \item \( s_{q-1} \) prefers \( p_{q-1} \) to \( p_{q-2} \),
    \item \( (s_{q-1}, p_{q-1}) \in M \setminus M' \), \( (s_{q-1}, p_{q-2}) \in M' \setminus M \),
    \item \( l_k \) offers both \( p_{q-1} \) and \( p_{q-2} \).
    \item \( l_k \) prefers \( s_{q-1} \) to \( s_{q-2} \),
\end{enumerate}

\noindent Since \( M' \) is stable, one of the following conditions hold:

\begin{enumerate}[label=(\roman*)]
    \item \( p_{q-1} \) is full in \( M' \), and \( l_k \) prefers the worst student in \( M'(p_{q-1}) \) to \( s_{q-1} \);
    \item \( p_{q-1} \) is undersubscribed in \( M' \), \( l_k \) is full, \( s_{q-1} \notin M'(l_k) \), and \( l_k \) prefers the worst student in \( M'(l_k) \) to \( s_{q-1} \).
\end{enumerate}

\noindent Since \( s_{q-1} \in M'(l_k) \), it follows from Proposition~\ref{prop:s-pref-l-prefs} that $p_{q-1}$ is full in $M'$and case~(i) holds. Since $(s_{q-1},p_{q-1}) \in M \setminus M'$, there exists some student, say $s_{q}$, such that $(s_q,p_{q-1}) \in M' \setminus M$; otherwise, $p_{q-1}$ would be oversubscribed in $M$. Furthermore, \( l_k \) prefers \( s_q \) to \( s_{q-1} \). Now, in $M$, \( s_q \) must be assigned to some project \( p_q \) such that $s_q$ prefers $p_q$ to $p_{q-1}$; otherwise \( (s_q, p_{q-1}) \) would block \( M \). Hence, \( (s_q, p_q) \in M \setminus M' \). Also, since \( s_q \in M'(l_k) \) and \( M(l_k) = M'(l_k) \), it follows that $l_k$ also offers \( p_q \). Thus, properties (1) - (4) hold for \( t = q \), completing the inductive step.

\medskip
\noindent It is easy to see that for each new student that we identify, $l_k$ prefers $s_{{t}}$ to $s_{{t - 1}}$ and prefers $s_{{t-1}}$ to $s_{{t-2}}$, $\dots$, and prefers $s_{2}$ to $s_{1}$, just as in Figure \ref{fig:infiniteseq}. Hence, all identified students must be distinct. Since the sequence of distinct students is infinite, we reach an immediate contradiction. This contradiction implies that $M(l_k) \neq M'(l_k)$, and the sequence must terminate with some student \( s_r \in M'(l_k) \setminus M(l_k) \) such that \( l_k \) prefers \( s_r \) to \( s_1 \). Hence, \( M(l_k) \neq M'(l_k) \), as required.

\begin{figure}[H]
\centering
\small
\begin{tabular}{llll}
& $M$ & $M'$  & \\ 
\hline
$s_1:$ & $p_1$ & $p_0$ \;&  $l_k$: \; $\dots$ \; $s_t$ \; $s_{t-1}$ \; $\dots$ \; $s_3$ \; $s_2$ \; $s_1$ \\
$s_2:$ & $p_2$ & $p_1$ \;& \\ 
$s_3:$ & $p_3$ & $p_2$ \;& \\ 
$\vdots$\; & $\vdots$ & $\vdots$  &  \\
$s_{t-1}:$ & $p_{t-1}$ &  $p_{t-2}$ \;& \\ 
$s_t:$ & $p_t$  &  $p_{t-1}$ \; & \\
$\vdots$\; & $\vdots$ & $\vdots$  &  \\ 
\end{tabular}
\caption{An illustration of the sequence of students generated in Lemma~\ref{lem:samelecturer1}, with $(s_r, p_r) \in M$ and $(s_r, p_{r-1}) \in M'$ for all $r \geq 2$}
\label{fig:infiniteseq} 
\vspace{-0.5cm}
\end{figure}
\end{proof}

\noindent As an example, consider the instance \( I_1 \) in Figure \ref{fig:example1}. Here, \( s_1 \) is assigned to \( p_1 \) in \( M_1 \) and \( p_2 \) in \( M_2 \), where both \( p_1 \) and \( p_2 \) are offered by \( l_1 \). By Lemma \ref{lem:samelecturer1}, we can identify some \( s_3 \in M_2(l_1) \setminus M_1(l_1) \) such that \( l_1 \) prefers \( s_3 \) to \( s_1 \).

\begin{lemma}
\label{lem:student-pref-lecturer-pref}
Let \( M \) and \( M' \) be two stable matchings in an instance \( I \), and let \( l_k \) be a lecturer such that \( M(l_k) \neq M'(l_k) \).  
If there exists a student \( s \in M(l_k) \setminus M'(l_k) \) who prefers \( M \) to \( M' \), then \( l_k \) prefers \( M' \) to \( M \).
\end{lemma}

\begin{proof}
Let $M$ and $M'$ be two stable matchings in \( \mathcal{M} \). Let \( l_k \) be some lecturer such that \( M(l_k) \neq M'(l_k) \), and let \( s_i \in M(l_k) \setminus M'(l_k)\) be some student who prefers  \( M \) to \( M' \). To prove that $l_k$ prefers $M'$ to $M$, we construct a one-to-one mapping
\[
f : M(l_k) \setminus M'(l_k) \to M'(l_k) \setminus M(l_k)
\]
such that for each student \( s \in M(l_k) \setminus M'(l_k) \) who prefers $M$ to $M'$, \( l_k \) prefers \( f(s) \) to \( s \). That is, for each such student in $M(l_k) \setminus M'(l_k)$, we can find a corresponding student $s' \in M'(l_k) \setminus M(l_k)$ such that $l_k$ prefers $s'$ to $s$.

\medskip
\noindent We say that a student $s_x \in M(l_k) \) is a \emph{dominated student} if \( l_k \) prefers all students in \( M'(l_k) \) to \( s_x \). There are two possible cases for the students in $M(l_k)$ who prefer $M$ to $M'$:

\medskip
\noindent \textbf{Case 1:} All such students are dominated.\\
In this case, $l_k$ prefers all students in $M'(l_k)$ to each such student in $M(l_k) \setminus M'(l_k)$. Since $|M'(l_k) \setminus M(l_k)| = |M(l_k) \setminus M'(l_k)|$, we can construct a one-to-one mapping for each student $s \in M(l_k) \setminus M'(l_k)$ who prefers $M$ to $M'$ to another student $s' \in M'(l_k) \setminus M(l_k)$ such that $l_k$ prefers $s'$ to $s$. In this way, the mapping is valid, and \( l_k \) prefers \( M' \) to \( M \).

\bigskip
\noindent \textbf{Case 2:} There exists at least one student in \( M(l_k) \setminus M'(l_k) \) who prefers $M$ to $M'$ and is not dominated.\\
Let \( s_1 \in M(l_k) \setminus M'(l_k) \) be a non-dominated student who prefers \( M \) to \( M' \), and let \( p_1 = M(s_1) \). It follows that $l_k$ prefers $s_1$ to at least one student in $M'(l_k)$. Since $M'$ is a stable matching, then either (i) or (ii) holds as follows:
\begin{enumerate}[label=(\roman*)]
    \item \( p_1 \) is full in \( M' \), and \( l_k \) prefers the worst student in \( M'(p_1) \) to \( s_1 \); or
    \item \( p_1 \) is undersubscribed in \( M' \), \( l_k \) is full in \( M' \), $s_1 \notin M'(l_k)$, and \( l_k \) prefers each student in \( M'(l_k) \) to \( s_1 \).
\end{enumerate}
\noindent Since $l_k$ prefers $s_1$ to at least one student in $M'(l_k)$, it follows from Proposition~\ref{prop:s-pref-l-prefs} that $p_1$ is full in $M'$. So case (i) holds, and there exists some student $s_2 \in M'(p_1) \setminus M(p_1)$ such that $l_k$ prefers $s_2$ to $s_1$. First suppose that $s_2$ prefers $M'$ to $M$. If $p_1$ is full in $M$, then $l_k$ prefers $s_2$ to some student in $M(p_1)$ (namely $s_1$), so $(s_2,p_1)$ blocks $M$. If $p_1$ is undersubscribed in $M$, then $l_k$ prefers $s_2$ to some student in $M(l_k)$ (namely $s_1$), and $(s_2,p_1)$ again blocks $M$ (this holds whether \(l_k\) is full or undersubscribed in \(M\)). Therefore, $s_2$ prefers $M$ to $M'$. 

\medskip
\noindent Let $S_k(M,M')$ denote the set of students assigned in $M$ and $M'$ to different projects offered by lecturer $l_k$. If \( s_2 \in M'(l_k) \setminus M(l_k) \), we define \( f(s_1) = s_2 \) and stop. Otherwise, $s_2 \in S_k(M,M')$. In this case, let $p_2 = M(s_2)$. Then $p_2$ is offered by $l_k$ and $(s_2,p_2) \in M\setminus M'$. Since $M'$ is a stable matching, then either (i) or (ii) holds as follows: 
\begin{enumerate}[label=(\roman*)]
    \item \( p_2 \) is full in \( M' \), and \( l_k \) prefers the worst student in \( M'(p_2) \) to \( s_2 \); or
    \item \( p_2 \) is undersubscribed in \( M' \), \( l_k \) is full in \( M' \), $s_2 \notin M'(l_k)$, and \( l_k \) prefers each student in \( M'(l_k) \) to \( s_2 \).
\end{enumerate}
Since \( l_k \) prefers \( s_2 \) to \( s_1 \), and prefers $s_1$ to at least one student in $M'(l_k)$, then \( l_k \) prefers $s_2$ to at least one student in $M'(l_k)$. By Proposition \ref{prop:s-pref-l-prefs}, it follows that $p_2$ is full in $M'$. Since \( s_2 \in M(p_2) \setminus M'(p_2) \) and \( p_2 \) is full in \( M' \), there exists some student \( s_3 \in M'(p_2) \setminus M(p_2) \); for otherwise, \( p_2 \) is oversubscribed in \( M \). Moreover, \( l_k \) prefers \( s_3 \) to \( s_2 \). Suppose that $s_3$ prefers $M'$ to $M$. If $p_2$ is full in $M$, then $l_k$ prefers $s_3$ to some student in $M(p_2)$ (namely $s_2$), so $(s_3,p_2)$ blocks $M$. If $p_2$ is undersubscribed in $M$, then $l_k$ prefers $s_3$ to some student in $M(l_k)$ (namely $s_2$), and $(s_3,p_2)$ again blocks $M$. Therefore, $s_3$ prefers $M$ to $M'$.  If $s_3 \in M'(l_k) \setminus M(l_k)$, we define $f(s_1) = s_3$ and stop. Otherwise, we continue this process to obtain a sequence of students $s_4, s_5, \dots, s_t$, where each student in the sequence prefers $M$ to $M'$ and is preferred by $l_k$ to their predecessor; that is, $l_k$ prefers $s_r$ to $s_{r-1}$ for $1 < r \leq t$. Since the number of students is finite, this sequence must eventually terminate with a student $s_t \in M'(l_k) \setminus M(l_k)$, at which point we define $f(s_1) = s_t$.

\medskip
\noindent The sequence $s_1, s_2, \dots, s_t$ is such that
\begin{itemize}
    \item $s_1 \in M(l_k) \setminus M'(l_k)$
    \item $s_r \in M(l_k) \cap M'(l_k)$ for $1 < r < t$
    \item $s_t \in M'(l_k) \setminus M(l_k)$
    \item $l_k$ prefers $s_r$ to $s_{r-1}$ for $1 < r \leq t$.
    \item $s_r$ prefers $M$ to $M'$ for $1 \leq r \leq t$.
\end{itemize}

\noindent We repeat this construction for every non-dominated student in \( M(l_k) \setminus M'(l_k) \). We also ensure that if the same project appears in the sequences starting from multiple students in $M(l_k) \setminus M'(l_k)$, then we can assign a distinct student from $M'(l_k) \setminus M(l_k)$ to that project in each case. Suppose some project $p_x \in P_k$ arises multiple times in $M \setminus M'$, each time assigned to a different student. Then, as argued earlier, $p_x$ must be full in $M'$, and all students assigned to $p_x$ in $M'$ are preferred by $l_k$ to the students it was assigned to in $M$. Since each occurrence of $p_x$ corresponds to a different student in $M(l_k) \setminus M'(l_k)$, and $p_x$ is full in $M'$, there are sufficiently many students in $M'(l_k) \setminus M(l_k)$ assigned to $p_x$ from which we can choose. We select a distinct one for each occurrence, preserving the one-to-one mapping.

\medskip
\noindent Finally, for the dominated students, we assign the remaining unassigned students in \( M'(l_k) \setminus M(l_k) \) arbitrarily. Since for each dominated student \( s \in M(l_k) \setminus M'(l_k) \), $l_k$ prefers each student in $M'(l_k)$ to $s$, the condition that $l_k$ prefers $f(s)$ to $s$ still holds. Thus, in both cases, we construct a valid one-to-one mapping from \( M(l_k) \setminus M'(l_k) \) to \( M'(l_k) \setminus M(l_k) \) such that $l_k$ prefers each student in \( M'(l_k) \setminus M(l_k) \) to the corresponding student in \( M(l_k) \setminus M'(l_k) \). Therefore, \( l_k \) prefers \( M' \) to \( M \), as required.
\end{proof}

\begin{lemma}
\label{lem:s-prefers-l-prefers1}
Let \( M \) and \( M' \) be two stable matchings in a given instance \( I \). Suppose some student \( s_i \) is assigned to different projects in \( M \) and \( M' \), and that in \( M' \), \( s_i \) is assigned to a project \( p_j \) offered by lecturer \( l_k \). Suppose further that \( s_i \) prefers \( M \) to \( M' \). Then:
\begin{enumerate}[label=(\alph*)]
    \item If there exists a student in \( M(p_j) \setminus M'(p_j) \), then \( l_k \) prefers \( s_i \) to each student in \( M(p_j) \setminus M'(p_j) \).
    \item If \( p_j \) is undersubscribed in \( M \), then \( l_k \) prefers \( s_i \) to each student in \( M(l_k) \setminus M'(l_k) \).
\end{enumerate}
\end{lemma}

\begin{proof}
Let $M$ and $M'$ be two stable matchings in $I$, and let $s_i$ be a student assigned to different projects in $M$ and $M'$, where $s_i$ prefers $M$ to $M'$. Let $p_j = M'(s_i)$, and let $l_k$ be the lecturer offering $p_j$. In Case (a), we show that if $M(p_j) \setminus M'(p_j)$ is non-empty, then there exists a student $s' \in M(p_j) \setminus M'(p_j)$ who prefers $M'$ to $M$. In Case~(b), we show that if $p_j$ is undersubscribed in $M$, then there exists a student $s' \in M(l_k)$ who is assigned to different projects in $M$ and $M'$, and who prefers $M'$ to $M$. In both cases, we use the student $s'$ identified to initiate an inductive argument. This produces a sequence of distinct students, where each student prefers $M'$ to $M$, and is preferred, by the lecturer they are assigned to in $M$, to the previous student in the sequence. 

\medskip
\noindent \textbf{Case (a):} Suppose there exists some student \( s' \in M(p_j) \setminus M'(p_j) \), and suppose for contradiction that \( l_k \) prefers \( s' \) to \( s_i \). Suppose \( p_j \) is full in \( M' \).
If \( s' \) prefers \( M \) to \( M' \), then \( (s', p_j) \) blocks \( M' \), since \( l_k \) prefers \( s' \) to some student in \( M'(p_j) \) (namely \( s_i \)), a contradiction. Now suppose that \( p_j \) is undersubscribed in \( M' \). Since \( l_k \) prefers \( s' \) to some student in \( M'(l_k) \) (again, \( s_i \)), then \( (s', p_j) \) blocks \( M' \), another contradiction. Therefore, \( s' \) prefers \( M' \) to \( M \).

\medskip
\noindent \textbf{Case (b):} Suppose \( p_j \) is undersubscribed in \( M \), and suppose for contradiction that \( l_k \) prefers each student in \( M(l_k) \setminus M'(l_k) \) to \( s_i \). Recall that $s_i \in M'(p_j) \setminus M(p_j)$. Since $p_j$ is undersubscribed in $M$ and \( |M(l_k)| = |M'(l_k)| \), there exists some project \( p' \in P_{k} \) and some student \( s' \in M(p') \setminus M'(p') \), where \( p' \) is undersubscribed in \( M' \).  First suppose that \( s' \) prefers \( M \) to \( M' \). If \( s' \in S_{k}(M, M') \), then \( (s', p') \) blocks \( M' \), a contradiction. Thus, \( s' \in M(l_k) \setminus M'(l_k) \). Since, by our assumption, \( l_k \) prefers each student in \( M(l_k) \setminus M'(l_k) \) to \( s_i \), it follows that \( l_k \) prefers \( s' \) to \( s_i \). However, the pair \( (s', p') \) again blocks \( M' \). Therefore, \( s' \) prefers \( M' \) to \( M \). 

\medskip
\noindent The remainder of the proof for Cases (a) and (b) proceeds in an identical manner. We therefore continue with the inductive step that satisfies the conditions of both cases in the following paragraph.

\medskip
\noindent Let $s'$ be the student, identified in either Case~(a) or Case~(b) above, who prefers $M'$ to $M$. Let $s_0 = s_i$, $l_0 = l_k$. Let $s_1 = s'$, $p_0 = M(s_1)$, and $p_1 = M'(s_1)$, and let $l_1$ be the lecturer who offers $p_1$. Note that it is possible that $p_0 = p_j$ and $l_1 = l_0$. We have $(s_1, p_1) \in M' \setminus M$, $(s_1, p_0) \in M \setminus M'$, and $l_0$ prefers $s_1$ to $s_0$. Moreover, $s_1$ prefers $M'$ to $M$, while $s_0$ prefers $M$ to $M'$.

\medskip
\noindent We now proceed by identifying students \( s_2, s_3, \dots \), projects \( p_2, p_3, \dots \), and lecturers \( l_2, l_3, \dots \), such that for each \( t \geq 2 \), the following properties hold:

\begin{enumerate}[label=(\arabic*)]
  \item \( s_t \) prefers \( M' \) to \( M \);
  \item \( (s_t, p_t) \in M' \setminus M \), where project \( p_t \) is offered by lecturer \( l_t \);
  \item  \( s_t \) is assigned in \( M \) to some project offered by lecturer \( l_{t-1} \), and:
  \begin{itemize}
      \item if \( p_{t-1} \) is full in \( M \), then \( (s_t, p_{t-1}) \in M \setminus M' \); 
      \item otherwise, there exists a project \( p'_{t-1}\) such that \( (s_t, p'_{t-1}) \in M \setminus M' \) (where $l_{t-1}$ offers both $p_{t-1}$ and $p'_{t-1}$).
  \end{itemize}

  \item Lecturer \( l_{t-1} \) prefers \( s_t \) to \( s_{t-1} \).
\end{enumerate}

\noindent \textbf{Base case (t = 2):}  Since $s_1$ prefers $M'$ to $M$, so $s_1$ prefers $p_1$ to $p_0$. Moreover, since $(s_1, p_1) \in M' \setminus M$ and $(s_1, p_0) \in M \setminus M'$, it follows that $p_1 \neq p_0$. Also, lecturer $l_0$ prefers $s_1$ to $s_0$. By the stability of $M$, one of the following two cases holds:
\begin{enumerate}[label=(\roman*)]
    \item \( p_1 \) is full in \( M \), and \( l_1 \) prefers the worst student in \( M(p_1) \) to \( s_1 \);
    \item \( p_1 \) is undersubscribed in \( M \), \( l_1 \) is full in \( M \), \( s_1 \notin M(l_1) \), and \( l_1 \) prefers the worst student in \( M(l_1) \) to \( s_1 \).
\end{enumerate}
\noindent Case (i): Since $p_1$ is full in $M$ and $(s_1, p_1) \notin M$, there exists another student $s_2$ such that $(s_2, p_1) \in M \setminus M'$, otherwise $p_1$ would be oversubscribed in $M$. Moreover, $l_1$ prefers $s_2$ to $s_1$. Clearly, $s_2$ is assigned in $M'$; let $p_2 = M'(s_2)$, and let $l_2$ be the lecturer who offers $p_2$. Then $(s_2, p_2) \in M' \setminus M$. If $s_2$ prefers $p_1$ to $p_2$, then the pair $(s_2, p_1)$ blocks $M'$, since $p_1$ is full in $M$, and $l_1$ prefers $s_2$ to $s_1 \in M'(p_1)$. Thus, $s_2$ must prefer $p_2$ to $p_1$, that is, $s_2$ prefers $M'$ to $M$. Note that $s_2 \neq s_1$, since $l_1$ prefers $s_2$ to $s_1$; and $s_2 \neq s_0$, since $s_0$ prefers $M$ to $M'$, while $s_2$ prefers $M'$ to $M$.

\medskip
\noindent Case (ii): Since $(s_1, p_1) \in M' \setminus M$ and $p_1$ is undersubscribed in $M$, there exists some project $p'_1$ offered by $l_1$, such that $p'_1$ is undersubscribed in $M'$; otherwise, $l_1$ would be oversubscribed in $M$. Thus, there exists some student $s_2$ such that $(s_2, p'_1) \in M \setminus M'$ and $l_1$ prefers $s_2$ to $s_1$. Moreover, $s_2$ is assigned in $M'$; let $p_2 = M'(s_2)$, and let $l_2$ be the lecturer who offers $p_2$. Then $(s_2, p_2) \in M' \setminus M$. If $s_2$ prefers $p'_1$ to $p_2$, then the pair $(s_2, p'_1)$ would block $M'$, since $p'_1$ is undersubscribed in $M'$ and $l_1$ prefers $s_2$ to $s_1$. Hence, $s_2$ prefers $p_2$ to $p'_1$, i.e., $s_2$ prefers $M'$ to $M$. As before, $s_2 \neq s_1$, since $l_1$ prefers $s_2$ to $s_1$, and $s_2 \neq s_0$, since $s_2$ prefers $M'$ to $M$, whereas $s_0$ prefers $M$ to $M'$.

\medskip
\noindent In both cases (i) and (ii), we have identified a student $s_2$ who prefers $M'$ to $M$, with $(s_2, p_2) \in M' \setminus M$, where $p_2$ is offered by lecturer $l_2$. Moreover, $s_2$ is assigned in $M$ to some project offered by $l_1$. If $p_1$ is full in $M$, then $(s_2,p_1) \in M \setminus M'$; otherwise, $(s_2,p'_1) \in M \setminus M'$. Moreover, $l_1$ prefers $s_2$ to $s_1$. Thus, properties (1)–(4) hold for $t = 2$, completing the base case.

\bigskip
\noindent \textbf{Inductive step:} Suppose properties~(1)-(4) hold for some $t = q - 1 \geq 2$, that is:
\begin{enumerate}[label=(\arabic*)]
  \item $s_{q-1}$ prefers $M'$ to $M$;
  \item $(s_{q-1}, p_{q-1}) \in M' \setminus M$, where $p_{q-1}$ is offered by lecturer $l_{q-1}$;
  \item $s_{q-1}$ is assigned in $M$ to some project offered by lecturer $l_{q-2}$, and:
  \begin{itemize}
      \item if $p_{q-2}$ is full in $M$, then $(s_{q-1}, p_{q-2}) \in M \setminus M'$;
      \item otherwise, there exists a project $p'_{q-2}$ such that $(s_{q-1}, p'_{q-2}) \in M \setminus M'$ (where $l_{q-2}$ offers both $p_{q-2}$ and $p'_{q-2}$).
  \end{itemize}
  \item Lecturer $l_{q-2}$ prefers $s_{q-1}$ to $s_{q-2}$.
\end{enumerate}

\noindent By the stability of $M$, one of the following two cases must hold:
\begin{enumerate}[label=(\roman*)]
  \item $p_{q-1}$ is full in $M$, and $l_{q-1}$ prefers the worst student in $M(p_{q-1})$ to $s_{q-1}$;
  \item $p_{q-1}$ is undersubscribed in $M$, $l_{q-1}$ is full in $M$, $s_{q-1} \notin M(l_{q-1})$, and $l_{q-1}$ prefers the worst student in $M(l_{q-1})$ to $s_{q-1}$.
\end{enumerate}

\noindent Case~(i): Since $p_{q-1}$ is full in $M$ and $(s_{q-1}, p_{q-1}) \notin M$, there exists a student $s_q$ such that $(s_q, p_{q-1}) \in M \setminus M'$, otherwise $p_{q-1}$ would be oversubscribed in $M$. Moreover, $l_{q-1}$ prefers $s_q$ to $s_{q-1}$. Clearly, $s_q$ is assigned in $M'$; let $p_q = M'(s_q)$, and let $l_q$ be the lecturer who offers $p_q$. Then $(s_q, p_q) \in M' \setminus M$. If $s_q$ prefers $p_{q-1}$ to $p_q$, then the pair $(s_q, p_{q-1})$ blocks $M'$, since $p_{q-1}$ is full in $M$ and $l_{q-1}$ prefers $s_q$ to $s_{q-1} \in M'(p_{q-1})$. Hence $s_q$ prefers $p_q$ to $p_{q-1}$, that is, $s_q$ prefers $M'$ to $M$.

\bigskip
\noindent Case~(ii): Since $(s_{q-1}, p_{q-1}) \in M' \setminus M$ and $p_{q-1}$ is undersubscribed in $M$, there must exist some project $p'_{q-1} \in P_{{q-1}}$ that is undersubscribed in $M'$, for otherwise $l_{q-1}$ would be oversubscribed in $M$. Then there exists a student $s_q$ such that $(s_q, p'_{q-1}) \in M \setminus M'$, and $l_{q-1}$ prefers $s_q$ to $s_{q-1}$. Moreover, $s_q$ is assigned in $M'$; let $p_q = M'(s_q)$ and let $l_q$ be the lecturer who offers $p_q$. Then $(s_q, p_q) \in M' \setminus M$. If $s_q$ prefers $p'_{q-1}$ to $p_q$, then the pair $(s_q, p'_{q-1})$ blocks $M'$, since $p'_{q-1}$ is undersubscribed in $M'$ and $l_{q-1}$ prefers $s_q$ to $s_{q-1} \in M'(l_{q-1})$. Hence $s_q$ prefers $p_q$ to $p'_{q-1}$, that is, $s_q$ prefers $M'$ to $M$.

\bigskip
\noindent In both cases (i) and (ii), we have identified a student $s_q$ who prefers $M'$ to $M$, with $(s_q, p_q) \in M' \setminus M$, where $p_q$ is offered by lecturer $l_q$. Moreover, $s_q$ is assigned in $M$ to some project offered by $l_{q-1}$. If $p_{q-1}$ is full in $M$, then $(s_q,p_{q-1}) \in M \setminus M'$; otherwise, $(s_q,p'_{q-1}) \in M \setminus M'$. Moreover, $l_{q-1}$ prefers $s_q$ to $s_{q-1}$. Thus, properties (1)–(4) hold for $t = q$, completing the base case.

\bigskip
\noindent It follows from the construction that each project $p_t$ differs from the previous project $p_{t-1}$, since each student $s_t$ is assigned to different projects in $M$ and $M'$. We now show that all students in the sequence $s_1, s_2, \dots$ are distinct, by induction on $t$. Clearly, $s_2 \ne s_1$, since the lecturer $l_1$ prefers $s_2$ to $s_1$, and $s_2 \ne s_0$, since $s_2$ prefers $M'$ to $M$ while $s_0$ prefers $M$ to $M'$. Now suppose, as the inductive hypothesis, that $s_1, \dots, s_t$ are all distinct for some $t \geq 2$, and suppose for contradiction that $s_{t+1} = s_q$ for some $1 \leq q \leq t$. In the construction, $s_{t+1}$ is selected by lecturer $l_t$ to prevent the pair $(s_t, p_t) \in M' \setminus M$ from blocking $M$. This means that either $(s_{t+1}, p_t) \in M \setminus M'$, if $p_t$ is full in $M$, or $(s_{t+1}, p'_t) \in M \setminus M'$, if $p_t$ is undersubscribed in $M$. In both cases, $l_t$ prefers $s_{t+1}$ to $s_t$.

\medskip
\noindent Since $s_{t+1} = s_q$, this student must have appeared earlier in the sequence and must have been selected to resolve a different blocking pair involving some earlier student $s_{q-1}$. In that case, $l_t$ is using the same student $s_q$ to resolve two different blocking pairs: one involving $(s_t, p_t)$ and one involving $(s_{q-1}, p_t)$. But by the inductive hypothesis, $s_t \ne s_{q-1}$, and so $l_t$ should have selected two different students. This contradicts the construction, which requires that each blocking pair is resolved by a student who has not already appeared in the sequence. Therefore, $s_{t+1} \ne s_q$ for all $1 \leq q \leq t$, and so the sequence $s_1, s_2, \dots$ consists of distinct students. Since the number of students is finite, the construction must eventually terminate. This establishes our claim and completes the proof.

\bigskip
\noindent To illustrate this proof, consider Figure \ref{fig:infiniteseq2}. Suppose that a student \( s_{0} \) prefers \( M \) to \( M' \), where in \( M' \), \( s_{0} \) is assigned to \( p_{0} \), a project offered by \( l_{0} \). Further, suppose there exists a student \( s_{1} \in M(p_{0}) \setminus M'(p_{0}) \), and suppose for a contradiction that \( l_{0} \) prefers \( s_{1} \) to \( s_{0} \). Clearly, if \( s_{1} \) also prefers \( M \) to \( M' \), then \( (s_{1}, p_{0}) \) blocks \( M' \). This implies that \( s_{1} \) prefers \( M' \) to \( M \). Let \( M'(s_{1}) \) be \( p_{1} \), where \( l_{1} \) offers $p_{1}$. To ensure the stability of \( M \), we continue identifying a sequence of distinct students, as illustrated in Figure \ref{fig:infiniteseq2}. However, since this sequence is infinite, we arrive at a contradiction.

\begin{figure}[H]
\centering
\small
\begin{tabular}{llll}
\hline
Students' preferences & Lecturers' preferences  & Offers \\ 
$s_0$: \; $p'_0$ \; $p_0$ \; &  $l_0$: \; $s_1$ \; $s_0$ & $p_0$ \\ 
$s_1$: \; $p_1$ \;  $p_0$ \; & $l_1$: \; $s_2$ \; $s_1$ & $p_1$ \\ 
$s_2$: \; $p_2$ \;  $p_1$ \; & $l_2$: \; $s_3$ \; $s_2$ & $p_2$ \\ 
$s_3$: \; $p_3$ \;  $p_2$ \; & $l_3$: \; $s_4$ \; $s_3$ & $p_3$ \\ 
\vdots \; \vdots \;  \vdots \; & \vdots &  \vdots \\ 
$s_t$: \; $p_t$ \;  $p_{t-1}$ \; & $l_t$: \; $s_{t+1}$ \; $s_t$ & $p_t$\\  
\vdots \; \vdots \;  \vdots \; & \vdots &  \vdots\\
\hline
\end{tabular}
\caption{A {\sc spa-s} instance illustrating the infinite sequence of students generated in Lemma~\ref{lem:s-prefers-l-prefers1}, where $(s_t, p_t) \in M'$ and $(s_t, p_{t-1}) \in M$.}
\label{fig:infiniteseq2}
\end{figure}
\vspace{-0.5cm}
\end{proof}

\begin{lemma}
\label{lem:samelecturer2}
Let \( M \) and \( M' \) be two stable matchings in \( I \). If some student \( s_i \) is assigned in \( M \) and \( M' \) to different projects offered by the same lecturer \( l_k \), and $s_i$ prefers $M$ to $M'$, then $l_k$ prefers $s_i$ to some student $s_z \in M(l_k) \setminus M'(l_k)$.
\end{lemma}
\begin{proof}
    Let $M$ and $M'$ be two stable matchings in $I$, and let \( s_1 \) be a student assigned to different projects in \( M \) and \( M' \), both offered by lecturer \( l_k \), where \( s_1 \) prefers \( M \) to \( M' \).  By Lemma~\ref{lem:samelecturer1}, there exists a student in \( M'(l_k) \setminus M(l_k) \) and, consequently, one in \( M(l_k) \setminus M'(l_k) \). Suppose, for a contradiction, that no student \( s \in M(l_k) \setminus M'(l_k) \) is worse than \( s_1 \) according to \( l_k \).  Let \( M(s_1)=p_0 \) and \( M'(s_1)=p_1 \), where \( s_1 \) prefers \( p_0 \) to \( p_1 \).  If \( p_1 \) is undersubscribed in \( M \), then by the second part of Lemma~\ref{lem:s-prefers-l-prefers1}, \( l_k \) prefers \( s_1 \) to each student in \( M(l_k) \setminus M'(l_k) \), a contradiction.  Hence \( p_1 \) must be full in \( M \).  Since \( (s_1,p_1) \in M' \setminus M \) and $p_1$ is full in $M$, there exists some \( (s_2,p_1) \in M \setminus M' \), and by the first part of Lemma~\ref{lem:s-prefers-l-prefers1}, \( l_k \) prefers \( s_1 \) to \( s_2 \).  

If \( s_2 \in M(l_k) \setminus M'(l_k) \), then \( l_k \) prefers \( s_1 \) to some student in \( M(l_k) \setminus M'(l_k) \), contradicting our assumption.  Hence \( s_2 \in S_k(M,M') \).  
Let \( M'(s_2)=p_2 \). First suppose that \( s_2 \) prefers \( p_2 \) to \( p_1 \), i.e. $s_2$ prefers $M'$ to $M$. Since \( s_1 \in M'(p_1) \setminus M(p_1) \), part (a) of Lemma~\ref{lem:s-prefers-l-prefers1} implies that \( l_k \) prefers \( s_2 \) to \( s_1 \), a contradiction.  
Thus \( s_2 \) prefers \( p_1 \) to \( p_2 \).  If \( p_2 \) is undersubscribed in \( M \), then by Lemma~\ref{lem:s-prefers-l-prefers1}, \( l_k \) prefers \( s_2 \) to each student in \( M(l_k) \setminus M'(l_k) \). Since \( l_k \) prefers \( s_1 \) to \( s_2 \),  it follows that $l_k$ prefers $s_1$ to student \( s \in M(l_k) \setminus M'(l_k) \), which again contradicts our assumption.  
Hence \( p_2 \) is full in \( M \).
Since \( (s_2,p_2) \in M' \setminus M \) and $p_2$ is full in $M$, there exists \( (s_3,p_2) \in M \setminus M' \), and by the first part of Lemma \ref{lem:s-prefers-l-prefers1}, \( l_k \) prefers \( s_2 \) to \( s_3 \).  

Again, if \( s_3 \in M(l_k) \setminus M'(l_k) \), then \( l_k \) prefers \( s_1 \) to \( s_2 \), and \( s_2 \) to some student in \( M(l_k) \setminus M'(l_k) \), contradicting our assumption.  
Thus \( s_3 \in S_k(M,M') \), and let \( M'(s_3) = p_3 \).  
Proceeding inductively as before, we obtain a sequence \( s_1, s_2, s_3, \ldots \) such that, for each \( t \ge 1 \), \( s_t \) prefers \( p_{t-1} \) to \( p_t \), \( (s_t,p_{t-1}) \in M \setminus M' \), \( (s_t,p_t) \in M' \setminus M \),
and both \( p_{t-1} \) and \( p_t \) are offered by \( l_k \), who prefers \( s_t \) to \( s_{t+1} \).  
Hence \( l_k \) prefers \( s_1 \) to \( s_2 \), \( s_2 \) to \( s_3 \), and so on, implying that all identified students are distinct.  
Since the number of students is finite, this sequence cannot continue indefinitely and must terminate with some \( s \in M(l_k) \setminus M'(l_k) \) such that \( l_k \) prefers \( s_1 \) to \( s \).
\end{proof}

\noindent The following corollary follows from Lemmas \ref{lem:student-pref-lecturer-pref} - \ref{lem:samelecturer2}:
\begin{corollary}
\label{cor:dominance-reversal}
Let $M$ and $M'$ be stable matchings in an instance $I$, and let $\preceq_S$ and $\preceq_L$ denote the student-oriented and lecturer-oriented dominance relations, respectively. 
Then $M \preceq_S M'$ if and only if $M' \preceq_L M$.
\end{corollary}
\begin{proof}
($\Rightarrow$)
Suppose $M \preceq_S M'$. Then each student either prefers $M$ to $M'$ or is indifferent between them. If $M=M'$, the claim is immediate. Otherwise, consider any lecturer $l_k$. If $M(l_k)=M'(l_k)$, then $l_k$ is indifferent between $M$ and $M'$, so $M' \preceq_L M$ holds for $l_k$. If $M(l_k)\neq M'(l_k)$, then there exist some student $s\in M(l_k)\setminus M'(l_k)$. 
Since $M \preceq_S M'$, $s$ prefers $M$ to $M'$.  By Lemma~\ref{lem:student-pref-lecturer-pref}, it follows that $l_k$ prefers $M'$ to $M$. Thus, for every $l_k$, $l_k$ is either indifferent or prefers $M'$ to $M$, i.e., $M'\preceq_L M$.

($\Leftarrow$) 
Conversely, suppose \( M' \preceq_L M \). 
Then each lecturer either prefers \( M' \) to \( M \) or is indifferent between them. Consider any student \( s \). If \( M(s) = M'(s) \), then \( s \) is indifferent between the two matchings, so $M \preceq_S M'$ holds for $s$.  Otherwise, let \( l_k \) be the lecturer offering \( M'(s) \), and suppose, for a contradiction, that \( s \) prefers \( M' \) to \( M \). If \( s \in S_k(M,M') \), then by Lemma~\ref{lem:samelecturer1}, \( l_k \) prefers some \( s_r \in M(l_k) \setminus M'(l_k) \) to \( s \); and by Lemma~\ref{lem:samelecturer2}, \( l_k \) prefers \( s \) to some \( s_z \in M'(l_k) \setminus M(l_k) \).  
Consequently, \( l_k \) prefers a student in \( M(l_k) \setminus M'(l_k) \) to one in \( M'(l_k) \setminus M(l_k) \), contradicting \( M' \preceq_L M \). If instead \( s \in M'(l_k) \setminus M(l_k) \), then by Lemma~\ref{lem:s-prefers-l-prefers1}, \( l_k \) prefers \( M \) to \( M' \), again a contradiction.  
Hence no student prefers \( M' \) to \( M \).  
Therefore each student is either indifferent between \( M \) and \( M' \) or prefers \( M \) to \( M' \), i.e. \( M \preceq_S M' \).
\end{proof}

\section{Stable Matchings in {\sc spa-s} Form a Distributive Lattice}
\label{sect:stablematchings-distlattice}
\noindent To show that $(\mathcal{M}, \preceq)$ forms a distributive lattice, we define the \emph{meet} and \emph{join} of any two stable matchings in $\mathcal{M}$ based on student preferences. Given two stable matchings $M$ and $M'$, the meet matching assigns each student to the project they prefer more between their projects in $M$ and $M'$, while the join matching assigns each student to the less preferred of the two. In Lemmas~\ref{lemma_meet_stable} and~\ref{lemma_join_stable}, we show that both the meet and join matchings are stable. These results show that the meet and join operations are well-defined in $\mathcal{M}$ and respect the dominance relation $\preceq$. Finally, in Theorem \ref{theorem-lattice}, we prove that the meet and join operations distribute, and that $(\mathcal{M},\preceq)$ is a distributive lattice.

\begin{definition}
\label{prop:meet_is_stable}
Let \( M \) and \( M' \) be two stable matchings in \( I \), and define a matching \( M^\land \) as follows: for each student \( s_i \), 
\begin{itemize}
    \item if \( s_i \) is unassigned in both \( M \) and \( M' \), then \( s_i \) is unassigned in \( M^\land \);
    \item if \( s_i \) is assigned to the same project in both \( M \) and \( M' \), then \( s_i \) is assigned to that project in \( M^\land \).
    \item otherwise, \( s_i \) is assigned in \( M^\land \) to the better of their projects in \( M \) and \( M' \).
\end{itemize}
\end{definition}

\noindent In Lemma \ref{lemma_meet_stable}, we prove that \( M^\land \) is a stable matching in \( I \). To show this, we first present Lemmas \ref{lemma_meet_lkunder} -- \ref{lemma_meet_matching}.

\begin{lemma}
\label{lemma_meet_lkunder}
If a lecturer \( l_k \) is undersubscribed in \( M^\land \), then \( l_k \) is undersubscribed in both \( M \) and \( M' \).
\end{lemma}

\begin{proof}
Suppose, for contradiction, that \( l_k \) is undersubscribed in \( M^\land \), but is full in both \( M \) and \( M' \). Then,
\(|M(l_k)| > |M^\land(l_k)|\) and \(|M'(l_k)| > |M^\land(l_k)|\).
Thus, there exists projects \(p_a, p_b \in P_k\) such that
\[
|M(p_a)| > |M^\land(p_a)| \quad \text{and} \quad |M'(p_b)| > |M^\land(p_b)|.
\]
Suppose that \(s_a \in M(p_a) \setminus M^\land(p_a)\) and \(s_b \in M'(p_b) \setminus M^\land(p_b)\). This implies that \(s_a \in M(p_a) \setminus M'(p_a)\) and \(s_b \in M'(p_b) \setminus M(p_b)\). By the construction of \(M^\land\), each student is assigned to the more preferred of their two projects in \(M\) and \(M'\). Therefore, (a) \(s_a\) prefers \(M'(s_a)\) to \(p_a\), and (b) \(s_b\) prefers \(M(s_b)\) to \(p_b\).  We claim that \(p_a\) is undersubscribed in \(M'\) and $p_b$ is undersubscribed in $M$. We now consider each of these cases in turn.

\bigskip
\noindent \textbf{Case (a):} \(s_a \in M(p_a) \setminus M'(p_a)\) and \(s_a\) prefers \(M'\) to \(M\). Suppose for contradiction that \(p_a\) is full in \(M'\). Since $p_a$ cannot be oversubscribed in $M$, we have \(|M'(p_a)| \geq |M(p_a)|\). Given that \(|M(p_a)| > |M^\land(p_a)|\), it follows that \(|M'(p_a)| > |M^\land(p_a)|\). Hence there exists some student \(s \in M'(p_a) \setminus M^\land(p_a)\), which means \(s \in M'(p_a) \setminus M(p_a)\). Moreover, by the construction of \(M^\land\), \(s\) prefers \(M\) to \(M'\).  Applying the first part of Lemma~\ref{lem:s-prefers-l-prefers1} to the matchings \(M\) and \(M'\), with \(s_a\) as a student who prefers \(M'\) to \(M\) and \(s \in M'(p_a) \setminus M(p_a)\), it follows that \(l_k\) prefers \(s_a\) to \(s\) (note that here, \(M\) and \(M'\) are swapped compared to Lemma~\ref{lem:s-prefers-l-prefers1}). On the other hand, applying the same lemma to \(M\) and \(M'\), with \(s\) as a student who prefers \(M\) to \(M'\) and \(s_a \in M(p_a) \setminus M'(p_a)\), it follows that \(l_k\) prefers \(s\) to \(s_a\). This yields a contradiction. Therefore, \(p_a\) is undersubscribed in \(M'\).

\bigskip
\noindent \textbf{Case (b):} Suppose \(s_b \in M'(p_b) \setminus M(p_b)\) and \(s_b\) prefers \(M\) to \(M'\). Following a similar argument to case (a), if, on the contrary, \(p_b\) were full in \(M\), then \(|M(p_b)| > |M^\land(p_b)|\), and there would exist some student \(s \in M(p_b) \setminus M^\land(p_b)\), and hence \(s \in M(p_b) \setminus M'(p_b)\). By the construction of \(M^\land\), \(s\) prefers \(M'\) to \(M\).  In this case, we have \(s_b \in M'(p_b) \setminus M(p_b)\) who prefers \(M\) to \(M'\), and \(s \in M(p_b) \setminus M'(p_b)\) who prefers \(M'\) to \(M\). Applying Lemma~\ref{lem:s-prefers-l-prefers1} to these two cases yields a contradiction on \(l_k\)'s preference list. Hence, \(p_b\) is undersubscribed in \(M\).

\bigskip
\noindent From cases (a) and (b), it follows that \(s_a\) prefers \(M'(s_a)\) to \(p_a\), where \(p_a\) is undersubscribed in \(M'\), and that \(s_b\) prefers \(M(s_b)\) to \(p_b\), where \(p_b\) is undersubscribed in \(M\). By the second part of Lemma~\ref{lem:s-prefers-l-prefers1}, this implies that \(l_k\) prefers \(s_a\) to each student in \(M'(l_k) \setminus M(l_k)\), and also prefers \(s_b\) to each student in \(M(l_k) \setminus M'(l_k)\). If \(s_b \in M'(l_k) \setminus M(l_k)\), then $l_k$ prefers $s_a$ to $s_b$. If instead \(s_b \in S_k(M,M')\), then since $s_b$ prefers $M$ to $M'$, Lemma~\ref{lem:samelecturer1} implies that there exists some \(s' \in M'(l_k) \setminus M(l_k)\) where \(l_k\) prefers \(s'\) to \(s_b\). Thus, \(l_k\) prefers \(s_a\) to \(s'\), and hence \(l_k\) prefers \(s_a\) to \(s_b\). On the other hand, since \(s_b\) prefers \(M\) to \(M'\) and \(p_b\) is undersubscribed in \(M\), then Lemma~\ref{lem:s-prefers-l-prefers1} implies that \(l_k\) prefers \(s_b\) to each student in \(M(l_k) \setminus M'(l_k)\). If \(s_a \in M(l_k) \setminus M'(l_k)\), then \(l_k\) prefers \(s_b\) to \(s_a\), a contradiction. If instead \(s_a \in S_k(M,M')\), then, since \(s_a\) prefers \(M'\) to \(M\), there exists some \(s \in M(l_k) \setminus M'(l_k)\) such that \(l_k\) prefers \(s\) to \(s_a\). Consequently, \(l_k\) prefers \(s_b\) to \(s\), and therefore to \(s_a\). This again yields a contradiction.

\medskip
\noindent Therefore, our assumption is false, and \( l_k \) is undersubscribed in both \( M \) and \( M' \).
\end{proof}
\begin{lemma}
\label{lemma_meet_pjunder}
If a project \( p_j \) is undersubscribed in \( M^\land \), then it is undersubscribed in at least one of \( M \) or \( M' \).
\end{lemma}

\begin{proof}
Let \( l_k \) be the lecturer who offers project \( p_j \). Suppose, for contradiction, that \( p_j \) is full in both \( M \) and \( M' \), but undersubscribed in \( M^\land \). Then \( |M(p_j)| > |M^\land(p_j)| \) and \( |M'(p_j)| > |M^\land(p_j)| \). It follows that there exists a student \( s_a \in M(p_j) \setminus M^\land(p_j) \). Since any student assigned to \( p_j \) in both \( M \) and \( M' \) must also be assigned to \( p_j \) in \( M^\land \), we conclude that \( s_a \notin M'(p_j) \), and hence \( s_a \in M(p_j) \setminus M'(p_j) \). Similarly, there exists a student \( s_b \in M'(p_j) \setminus M^\land(p_j) \), which implies that \( s_b \in M'(p_j) \setminus M(p_j) \). By the construction of \( M^\land \), \( s_a \) prefers \( M' \) to $M$, and \( s_b \) prefers \( M \) to \(M' \).

\medskip
\noindent By applying the first part of Lemma~\ref{lem:s-prefers-l-prefers1} to the matchings \( M' \) and \( M \), with \( s_a \) as a student who prefers \( M' \) to \( M \), and with \( s_b \in M'(p_j) \setminus M(p_j) \), it follows that \( l_k \) prefers \( s_a \) to \( s_b \) (note that here, \( M \) and \( M' \) are swapped compared to Lemma~\ref{lem:s-prefers-l-prefers1}). Conversely, applying the same lemma to \( M \) and \( M' \), with \( s_b \) as a student who prefers \( M \) to \( M' \) and \( s_a \in M(p_j) \setminus M'(p_j) \), we conclude that \( l_k \) prefers \( s_b \) to \( s_a \). This is a contradiction, since \( l_k \) cannot simultaneously prefer \( s_a \) to \( s_b \) and \( s_b \) to \( s_a \). Therefore, \( p_j \) must be undersubscribed in at least one of \( M \) or \( M' \).
\end{proof}

\begin{lemma}
\label{lemma_meet_matching}
\( M^\land \) is a matching.
\end{lemma}

\begin{proof}
By construction, no student is assigned to more than one project in \( M^\land \). It remains to show that no project or lecturer is oversubscribed in \( M^\land \). Suppose, for contradiction, that some project \( p_j \) is oversubscribed in \( M^\land \). Let \( l_k \) be the lecturer who offers \( p_j \). Then \( |M^\land(p_j)| > |M(p_j)| \) and \( |M^\land(p_j)| > |M'(p_j)| \), since both \( M \) and \( M' \) are valid matchings. Thus, there exist students \(s_a \in M^\land(p_j) \setminus M'(p_j)\) and \(s_b \in M^\land(p_j) \setminus M(p_j)\). It follows that \(s_a \in M(p_j) \setminus M'(p_j)\) and \(s_b \in M'(p_j) \setminus M(p_j)\), where \(s_a\) prefers \(M\) to \(M'\) and \(s_b\) prefers \(M'\) to \(M\).

\medskip
\noindent By stability of $M'$ and since $s_a$ prefers $p_j$ to $M'(s_a)$, it follows that $l_k$ prefers the worst student in $M'(p_j)$ to $s_a$ (if $p_j$ is full in $M'$) or the worst student in $M'(l_k)$ to $s_a$ (if $p_j$ is undersubscribed in $M'$). This implies that $l_k$ prefers $s_b$ to $s_a$, since $s_b \in M'(p_j)$. On the other hand, $s_b$ prefers $p_j$ to $M(s_b)$. If $p_j$ is full in $M$, then $l_k$ prefers $s_b$ to some student in $M(p_j)$ (namely $s_a$), and $(s_b,p_j)$ blocks $M$, a contradiction. If $p_j$ is undersubscribed in $M$, then $l_k$ prefers $s_b$ to some student in $M(l_k)$ (namely $s_a$), and $(s_b,p_j)$ blocks $M$, a contradiction (This holds whether $l_k$ is full or undersubscribed in $M$). Therefore, our assumption is false and no project is oversubscribed in \( M^\land \).

\medskip
\noindent Next, suppose for contradiction that some lecturer \( l_k \) is oversubscribed in \( M^\land \). Then there exist projects \( p_a \) and \( p_b \) offered by \( l_k \) such that \( |M^\land(p_a)| > |M'(p_a)| \) and \( |M^\land(p_b)| > |M(p_b)| \). Since both \( M \) and \( M' \) are valid matchings, this implies that \( p_a \) is undersubscribed in \( M' \) and \( p_b \) is undersubscribed in \( M \). Moreover, as established earlier, no project is oversubscribed in \( M^\land \). Let \( s_a \in M^\land(p_a) \setminus M'(p_a) \), so \( s_a \in M(p_a) \setminus M'(p_a) \), and let \( s_b \in M^\land(p_b) \setminus M(p_b) \), so \( s_b \in M'(p_b) \setminus M(p_b) \). By the definition of \( M^\land \), each student is assigned to the more preferred of their two projects in \( M \) and \( M' \); therefore, \( s_a \) prefers \( M \) to \( M' \), and \( s_b \) prefers \( M' \) to \( M \).

\medskip
\noindent Since $p_a$ is undersubscribed in $M'$, we have $s_a \notin M'(l_k)$; otherwise $(s_a,p_a)$ would block $M'$, regardless of whether $l_k$ is full or undersubscribed in $M'$.  Similarly, since \(p_b\) is undersubscribed in \(M\), we have \(s_b \notin M(l_k)\), otherwise \((s_b,p_b)\) would block \(M\).  Hence \(s_a \in M(l_k) \setminus M'(l_k)\) and \(s_b \in M'(l_k) \setminus M(l_k)\).  Now, by Lemma~\ref{lem:student-pref-lecturer-pref}, since \( s_a \) prefers \( M \) to \( M' \), it follows that \( l_k \) prefers \( M' \) to \( M \). Conversely, since \( s_b \) prefers \( M' \) to \( M \), the same lemma implies that \( l_k \) prefers \( M \) to \( M' \). This yields a contradiction. Hence our assumption is false, and no lecturer is oversubscribed in $M^\land$. Therefore, $M^\land$ is a valid matching.
\end{proof}

\begin{lemma}
\label{lemma_meet_stable}
\( M^\land \) is a stable matching.
\end{lemma}

\begin{proof}
Suppose for contradiction that \( (s, p) \) is a blocking pair for \( M^\land \), where project \( p \) is offered by lecturer \( l \). Then either:

\begin{itemize}
    \item[(S1)] \( s \) is unassigned in \( M^\land \), or
    \item[(S2)] \( s \) is assigned in \( M^\land \), but prefers \( p \) to $M^\land(s)$.
\end{itemize}

\noindent And one of the following four conditions holds for \( p \) and \(l\):

\begin{itemize}
    \item[(P1)] both \( p \) and \( l \) are undersubscribed in \( M^\land \).
    \item[(P2)] \( p \) is undersubscribed in \( M^\land \), \( l \) is full in \( M^\land \), and \( s \in M^\land(l) \).
    \item[(P3)] \( p \) is undersubscribed in \( M^\land \), \( l \) is full in \( M^\land \), and \( l \) prefers \( s \) to the worst student in \( M^\land(l) \).
    \item[(P4)] \( p \) is full in \( M^\land \), and \( l \) prefers \( s \) to the worst student in \( M^\land(p) \).
\end{itemize}

\noindent We consider each combination of conditions in turn. Note that the case (S1 \& P2) cannot arise, since $s$ is unassigned in $M^\land$ and therefore cannot belong to $M^\land(l)$.

\bigskip
\noindent \textbf{(S1 \& P1) and (S2 \& P1):} First, suppose $s$ is unassigned in $M^\land$. Then, by the construction of $M^\land$, $s$ is unassigned in both $M$ and $M'$. Alternatively, if $s$ is assigned in $M^\land$ and prefers $p$ to $M^\land(s)$, then $s$ prefers $p$ to both $M(s)$ and $M'(s)$, since $s$  receives their preferred project in $M^\land$. Now consider condition (P1), where both $p$ and its lecturer $l$ are undersubscribed in $M^\land$. By Lemma~\ref{lemma_meet_lkunder}, since $l$ is undersubscribed in $M^\land$, it is undersubscribed in both $M$ and $M'$. By Lemma~\ref{lemma_meet_pjunder}, since $p$ is undersubscribed in $M^\land$, it is undersubscribed in at least one of $M$ or $M'$. Without loss of generality, suppose $p$ is undersubscribed in $M'$. Then, whether $s$ is unassigned in $M'$, or assigned to a project they prefer less than $p$, the pair $(s, p)$ blocks $M'$, a contradiction. We conclude that  (S1 \& P1) and (S2 \& P1) cannot arise.\\

\noindent \textbf{(S2 \& P2), (S1 \& P3) and (S2 \& P3):}  
First, consider (S1), where \(s\) is unassigned in \(M^\land\). By construction of \(M^\land\), this means that \(s\) is unassigned in both \(M\) and \(M'\).  Next, consider (S2), where \(s\) is assigned in \(M^\land\) and prefers \(p\) to \(M^\land(s)\). Since each student in \(M^\land\) receives the more preferred of their two projects from \(M\) and \(M'\), it follows that \(s\) prefers \(p\) to both \(M(s)\) and \(M'(s)\). In conditions (P2) and (P3), \(p\) is undersubscribed in \(M^\land\). Hence, by Lemma~\ref{lemma_meet_pjunder}, \(p\) must be undersubscribed in at least one of \(M\) or \(M'\). 

\medskip
\noindent Suppose first that \(p\) is undersubscribed in both $M$ and $M'$. From (S2 \& P2), we have that \(s \in M^\land(l)\), which means that \(s \in M(l)\) or \(s \in M'(l)\).  If $s \in M'(l)$, then $(s,p)$ blocks $M'$, since $s$ prefers $p$ to $M'(s)$ and $p$ is undersubscribed in $M'$. This blocking pair arises whether $l$ is undersubscribed or full in $M'$. A similar contradiction arises in $M$ if \(s \in M(l)\). In conditions (S1 \& P3) and (S2 \& P3), let \(s_z\) be the worst student in \(M^\land(l)\), where \(l\) prefers \(s\) to \(s_z\).  If \(s_z \in M'(l)\), then \((s,p)\) blocks \(M'\), since \(s\) is either unassigned in \(M'\) or prefers \(p\) to \(M'(s)\), \(p\) is undersubscribed in \(M'\) and $l$ prefers $s$ to $s_z$ (again, this holds whether \(l\) is full or undersubscribed in $M'$).  A similar contradiction arises  in $M$ if \(s_z \in M(l)\).  

\medskip
\noindent Next, suppose without loss of generality that \(p\) is full in \(M\) but undersubscribed in \(M^\land\) and \(M'\).  If \(l\) is undersubscribed in \(M'\), then the pair \((s,p)\) blocks \(M'\), since \(s\) is either unassigned in $M'$ or prefers \(p\) to \(M'(s)\), and both \(p\) and \(l\) are undersubscribed in \(M'\).  Hence \(l\) is full in \(M'\), and therefore also full in \(M\). From (S2 \& P2), we have that \(s \in M(l)\) or \(s \in M'(l)\). If \(s \in M'(l)\), then \((s,p)\) blocks \(M'\), since \(s\) prefers \(p\) to \(M'(s)\), \(p\) is undersubscribed in \(M'\) and $l$ is full in $M'$.  Therefore, \(s \in M(l) \setminus M'(l)\). In (S1 \& P3) and (S2 \& P3), we have that \(s_z \in M(l)\) or \(s_z \in M'(l)\). If \(s_z \in M'(l)\), then \((s,p)\) blocks \(M'\), since \(s\) is either unassigned in $M'$ or prefers \(p\) to \(M'(s)\), \(p\) is undersubscribed in \(M'\), $l$ is full in $M'$, and \(l\) prefers \(s\) to \(s_z\).  Thus \(s_z \in M(l) \setminus M'(l)\).  

\bigskip
\noindent Now since $p$ is full in $M$ but undersubscribed in $M^\land$, it follows that $|M(p)| > |M^\land(p)|$. Since $l$ is full in both $M^\land$ and $M$, and $|M(p)| > |M^\land(p)|$, there exists a project $p_a$ offered by $l$ such that $|M^\land(p_a)| > |M(p_a)|$, implying that $p_a$ is undersubscribed in $M$. Thus, there exists a student $s_a \in M^\land(p_a) \setminus M(p_a)$, and hence $s_a \in M'(p_a) \setminus M(p_a)$. Since $s_a$ receives their more preferred project in $M^\land$, they prefer $p_a$ to $M(s_a)$. Moreover by the stability of $M$ (and since $p_a$ is undersubscribed in $M$), $l$ prefers the worst student in $M(l)$ to $s_a$. In the (S2 \& P2) case, it follows that $l$ prefers $s$ to $s_a$, since $s \in M(l) \setminus M'(l)$. However, since $s$ prefers $p$ to $M'(s)$, $p$ is undersubcribed in $M'$, $l$ is full in $M'$, and $l$ prefers $s$ to $s_a$, the pair $(s,p)$ blocks $M'$, a contradiction. In the (S1 \& P3) and (S2 \& P3) case, it follows that $l$ prefers $s_z$ to $s_a$, since $s_z \in M(l) \setminus M'(l)$. Consequently, $l$ prefers $s$ to $s_a$, since $l$ prefers $s$ to $s_z$. Again, we arrive at a similar contradiction as in (S2 \& P2) case, whereby $(s,p)$ blocks $M'$.

\medskip
\noindent Therefore, no blocking pair of type (S2 \& P2), (S1 \& P3) and (S2 \& P3) exists in \( M^\land \).\\

\noindent \textbf{(S1 \& P4) and (S2 \& P4):} Clearly, if $s$ is unassigned in $M^\land$, then by construction, $s$ is unassigned in both $M$ and $M'$. If instead $s$ is assigned in $M^\land$ and prefers $p$ to $M^\land(s)$, then $s$ prefers $p$ to both $M(s)$ and $M'(s)$, since $s$ receives the more preferred of their two projects in $M^\land$. Consider condition (P4), where $p$ is full in $M^\land$ and $l$ prefers $s$ to the worst student in $M^\land(p)$. Let $s_z$ be the worst student in $M^\land (p)$. Clearly, either $s_z \in M(p)$ or $s_z \in M'(p)$.

\medskip
\noindent \noindent First suppose \( (s_z, p) \in M \). If \( s \) is unassigned in \( M \), or if $s$ prefers \( p \) to \( M(s) \), then \( (s, p) \) blocks \( M \), since \( p \) is full and \( l \) prefers \( s \) to \( s_z \in M(p) \). This contradicts the stability of \( M \). A similar argument applies if \( (s_z, p) \in M' \): whether \( s \) is unassigned in \( M' \), or prefers \( p \) to \( M'(s) \), the pair \( (s, p) \) blocks \( M' \), again a contradiction. Therefore, no blocking pair of type (S1 \& P4) or (S2 \& P4) can exist in \( M^\land \).

\bigskip
\noindent In all possible cases, we arrive at a contradiction. Therefore, no blocking pair exists in \( M^\land \), and \( M^\land \) is stable.
\end{proof}

\medskip
\noindent We denote by $M \land M'$ the set of (student, project) pairs in which each student is assigned the better of her project in $M$ and $M'$; and it follows from Lemma $\ref{lemma_meet_stable}$ that $M \land M'$ is a stable matching. Hence, if each student is given the better of her project in any fixed set of stable matchings, then the resulting assignment is a stable matching. For the case where \(\mathcal{M}\) is the set of all stable matchings in \(I\), we denote by \(\bigwedge_{M \in \mathcal{M}} M\), or simply \(\bigwedge \mathcal{M}\), the resulting stable matching. This matching is student-optimal and, by Corollary~\ref{cor:dominance-reversal}, lecturer-pessimal.

\medskip
\begin{definition}
\label{prop:join_is_stable}
Let $M$ and $M'$ be two stable matchings in $I$, and define a matching $M^\lor$ as follows: for each student $s_i$,
\begin{itemize}
\item if $s_i$ is unassigned in both $M$ and $M'$, then $s_i$ is unassigned in $M^\lor$;
\item if $s_i$ is assigned to the same project in both $M$ and $M'$, then $s_i$ is assigned to that project in $M^\lor$;
\item otherwise, $s_i$ is assigned in $M^\lor$ to the worse of their two projects in $M$ and $M'$.
\end{itemize}  
\end{definition}

\noindent In Lemma \ref{lemma_join_stable}, we prove that $M^\lor$ is a stable matching in $I$. To prove this, we first present Lemmas~\ref{lemma_join_lkunder} -- \ref{lemma_join_matching}.

\begin{lemma}
\label{lemma_join_lkunder}
If a lecturer \( l_k \) is undersubscribed in \( M^\lor \), then \( l_k \) is undersubscribed in both \( M \) and \( M' \).
\end{lemma}

\begin{proof}
Suppose, for contradiction, that \( l_k \) is undersubscribed in \( M^\lor \), but is full in both \( M \) and \( M' \). Then
\(|M(l_k)| > |M^\lor(l_k)| \) and \( |M'(l_k)| > |M^\lor(l_k)|\). It follows that there exist students \( s_a \in M(l_k) \setminus M^\lor(l_k) \) and \( s_b \in M'(l_k) \setminus M^\lor(l_k) \). Hence, \( s_a \in M(l_k) \setminus M'(l_k) \), and \( s_b \in M'(l_k) \setminus M(l_k) \). By construction of \( M^\lor \), each student is assigned to the less preferred of their two projects in \( M \) and \( M' \); therefore, \( s_a \) prefers \( M \) to \( M' \), and \( s_b \) prefers \( M' \) to \( M \).

\medskip
\noindent By Lemma~\ref{lem:student-pref-lecturer-pref}, since \( s_a \) prefers \( M \) to \( M' \), it follows that \( l_k \) prefers \( M' \) to \( M \). Conversely, since \( s_b \) prefers \( M' \) to \( M \), the same lemma implies that \( l_k \) prefers \( M \) to \( M' \). This yields a contradiction.  Therefore, our assumption is false, and \( l_k \) must be undersubscribed in both \( M \) and \( M' \).
\end{proof}

\begin{lemma}
\label{lemma_join_pjunder}
If a project \( p_j \) is undersubscribed in \( M^\lor \), then it must be undersubscribed in at least one of \( M \) or \( M' \).
\end{lemma}

\begin{proof}
Suppose, for contradiction, that \( p_j \) is undersubscribed in \( M^\lor \), but full in both \( M \) and \( M' \). Then we have
\[
|M(p_j)| > |M^\lor(p_j)| \quad \text{and} \quad |M'(p_j)| > |M^\lor(p_j)|.
\]
It follows that there exist students \( s_a \in M(p_j) \setminus M^\lor(p_j) \) and \( s_b \in M'(p_j) \setminus M^\lor(p_j) \). In particular, \( s_a \in M(p_j) \setminus M'(p_j) \) and \( s_b \in M'(p_j) \setminus M(p_j) \). Since each student is assigned in \( M^\lor \) to the less preferred of their projects in \( M \) and \( M' \), it follows that \( s_a \) prefers \( M \) to \( M' \), and \( s_b \) prefers \( M' \) to \( M \).

\medskip
\noindent Let $l_k$ be the lecturer who offers $p_j$. By stability of $M'$, since $s_a$ prefers $p_j$ to $M'(s_a)$ and $p_j$ is full in $M'$, it follows that $l_k$ prefers the worst student in $M'(p_j)$ to $s_a$. In particular, $l_k$ prefers student $s_b \in M'(p_j)$ to $s_a$. However, since $s_b$ prefers $p_j$ to $M(s_b)$, $p_j$ is full in $M$, and $l_k$ prefers $s_b$ to some student in $M(p_j)$ (namely $s_a$), it follows that $(s_b,p_j)$ blocks $M$, a contradiction. Therefore, $p_j$ must be undersubscribed in at least one of $M$ or $M'$.
\end{proof}

\begin{lemma}
\label{lemma_join_matching}
\( M^\lor \) is a matching.
\end{lemma}

\begin{proof}
By construction, no student is assigned to more than one project in \( M^\lor \). It remains to show that no project or lecturer is oversubscribed in \( M^\lor \). Suppose, for contradiction, that some project \( p_j \) is oversubscribed in \( M^\lor \). Then
\[
|M^\lor(p_j)| > |M(p_j)| \quad \text{and} \quad |M^\lor(p_j)| > |M'(p_j)|.
\]
Thus, there exist students \( s_a \in M^\lor(p_j) \setminus M'(p_j) \) and \( s_b \in M^\lor(p_j) \setminus M(p_j) \). Since any student assigned to \( p_j \) in both \( M \) and \( M' \) would also be assigned to \( p_j \) in \( M^\lor \), it follows that \( s_a \in M(p_j) \setminus M'(p_j) \) and \( s_b \in M'(p_j) \setminus M(p_j) \). Moreover, \( s_a \) prefers \( M' \) to \( M \), and \( s_b \) prefers \( M \) to \( M' \), since each student is assigned in \( M^\lor \) to the less preferred of their two projects. Let \( l_k \) be the lecturer who offers \( p_j \). By the first part of Lemma~\ref{lem:s-prefers-l-prefers1}, since \( s_a \) prefers \( M' \) to \( M \) and \( s_b \in M'(p_j) \setminus M(p_j) \), it follows that \( l_k \) prefers \( s_a \) to \( s_b \). Similarly, since \( s_b \) prefers \( M \) to \( M' \) and \( s_a \in M(p_j) \setminus M'(p_j) \), it follows that \( l_k \) prefers \( s_b \) to \( s_a \). Thus, \( l_k \) simultaneously prefers \( s_a \) to \( s_b \), and \( s_b \) to \( s_a \), a contradiction. Therefore, no project is oversubscribed in \( M^\lor \).

\medskip
\noindent Next, suppose that there exists some lecturer \( l_k \) who is oversubscribed in \( M^\lor \). Then there must be some project \( p_a \in P_k \) such that \( |M^\lor(p_a)| > |M'(p_a)| \), meaning that \( p_a \) is undersubscribed in \( M' \), since no project can be oversubscribed in \( M^\lor \). Similarly, there exists some project \( p_b \in P_k \) such that \( |M^\lor(p_b)| > |M(p_b)| \), meaning that \( p_b \) is undersubscribed in \( M \).  Let \( s_a \) be a student such that \( (s_a, p_a) \in M^\lor \setminus M' \). Thus, \( s_a \in M(p_a) \setminus M'(p_a) \) and $s_a$ prefers $M'$ to $M$. Let \( s_b \) be a student such that \( (s_b, p_b) \in M^\lor \setminus M \); thus \( s_b \in M'(p_b) \setminus M(p_b) \) and $s_b$ prefers $M$ to $M'$.

\medskip
\noindent By Lemma~\ref{lem:s-prefers-l-prefers1}, since \( s_a \) prefers \( M' \) to \( M \) and \( p_a \) is undersubscribed in \( M' \), \( l_k \) prefers \( s_a \) to each student in \( M'(l_k) \setminus M(l_k) \) (Note that here, \( M \) and \( M' \) are swapped compared to the statement of Lemma~\ref{lem:s-prefers-l-prefers1}). If \(s_b \in M'(l_k) \setminus M(l_k)\), then \(l_k\) prefers \(s_a\) to \(s_b\).  If instead \(s_b \in S_k(M,M')\), then since \(s_b\) prefers \(M\) to \(M'\), Lemma~\ref{lem:samelecturer1} implies that there exists some \(s' \in M'(l_k) \setminus M(l_k)\) such that \(l_k\) prefers \(s'\) to \(s_b\). It follows that \(l_k\) prefers \(s_a\) to \(s'\), and consequently to \(s_b\). 

\medskip
\noindent On the other hand, by Lemma~\ref{lem:s-prefers-l-prefers1} again, since \( s_b \) prefers \( M \) to \( M' \) and \( p_b \) is undersubscribed in \( M \), it follows that \( l_k \) prefers \( s_b \) to each student in \(M(l_k) \setminus M'(l_k) \). If $s_a \in M(l_k) \setminus M'(l_k)$, then $l_k$ prefers $s_b$ to $s_a$, a contradiction. If instead $s_a \in S_k(M,M')$, then by Lemma \ref{lem:samelecturer1}, there exists some $s \in M(l_k) \setminus M'(l_k)$ where $l_k$ prefers $s$ to $s_a$ (Note that here $s_a$ prefers $M'$ to $M$). This implies that $l_k$ prefers $s_b$ to $s$, and consequently to $s_a$. This yields a contradiction on $l_k$'s preferences. Hence, no lecturer is oversubscribed in \( M^\lor \). Therefore, \( M^\lor \) is a matching.
\end{proof}

\begin{lemma}
\label{lemma_join_stable}
\( M^\lor \) is a stable matching.
\end{lemma}

\begin{proof}
Suppose for contradiction that \( (s, p) \) is a blocking pair for \( M^\lor \), where project \( p \) is offered by lecturer \( l \). Then either:

\begin{itemize}
    \item[(S1)] \( s \) is unassigned in \( M^\lor \), or
    \item[(S2)] \( s \) is assigned in \( M^\lor \), but prefers \( p \) to \( M^\lor(s) \).
\end{itemize}

\noindent And one of the following four conditions holds for the \( p \) and \( l \):

\begin{itemize}
    \item[(P1)] both \( p \) and \( l \) are undersubscribed in \( M^\lor \);
    \item[(P2)] \( p \) is undersubscribed in $M^\lor$, \( l \) is full in $M^\lor$, and \( s \in M^\lor(l) \);
    \item[(P3)] \( p \) is undersubscribed in $M^\lor$, \( l \) is full in $M^\lor$, and \( l \) prefers \( s \) to the worst student in \( M^\lor(l) \);
    \item[(P4)] \( p \) is full in $M^\lor$, and \( l \) prefers \( s \) to the worst student in \( M^\lor(p) \).
\end{itemize}

\noindent We consider each combination of conditions in turn. Note that the case (S1 \& P2) cannot arise, since $s$ is unassigned in $M^\lor$ and therefore cannot belong to $M^\lor(l)$.

\medskip
\noindent \textbf{(S1 \& P1):} Suppose \( s \) is unassigned in \( M^\lor \), and both \( p \) and \( l \) are undersubscribed in \( M^\lor \). By construction, if \( s \) is unassigned in \( M^\lor \), then \( s \) is unassigned in both \( M \) and \( M' \). By Lemma~\ref{lemma_join_lkunder}, since \( l \) is undersubscribed in \( M^\lor \), it is undersubscribed in both \( M \) and \( M' \). By Lemma~\ref{lemma_join_pjunder}, since \( p \) is undersubscribed in \( M^\lor \), it is undersubscribed in at least one of \( M \) or \( M' \). Without loss of generality, suppose \( p \) is undersubscribed in \( M' \). Then \( s \) is unassigned in \( M' \), and both $p$ and $l$ are undersubscribed in $M'$, the pair \( (s, p) \) blocks \( M' \), a contradiction. A similar argument applies if $p$ is undersubscribed in $M$. We conclude that no blocking pair of type (S1 \& P1) exists in \( M^\lor \).\\

\noindent \textbf{(S2 \& P1):} Suppose \( s \) is assigned in \( M^\lor \), prefers \( p \) to \( M^\lor(s) \), and both \( p \) and \( l \) are undersubscribed in \( M^\lor \). Since \( l \) is undersubscribed in \( M^\lor \), it follows from Lemma~\ref{lemma_join_lkunder} that \( l \) is undersubscribed in both \( M \) and \( M' \). In addition, Lemma~\ref{lemma_join_pjunder} implies that \( p \) is undersubscribed in at least one of \( M \) or \( M' \). Without loss of generality, assume that \( p \) is undersubscribed in \( M' \). Let \( p^* = M^\lor(s) \). Since \( s \) is assigned in \( M^\lor \), they must be assigned in at least one of \( M \) or \( M' \), and $p^*$ is the less preferred of the two. There are two possibilities for \( M(s) \) and \( M'(s) \):

\begin{itemize}
    \item[(i)] \( M'(s) = p^* \). Then $s$ prefers $p$ to $M'(s)$. Since \( s \) prefers \( p \) to \( p^* \), and both \( p \) and \( l \) are undersubscribed in \( M' \), the pair \( (s, p) \) blocks \( M' \), a contradiction.

    \item[(ii)] \( M(s) = p^* \). Then \( s \in M(p^*) \setminus M'(p^*) \). Suppose first that $p$ is undersubscribed in $M$. Then, since $s$ prefers $p$ to $p^*$, and both $p$ and $l$ are undersubscribed in $M$, the pair $(s, p)$ blocks $M$, a contradiction. Next suppose that $p$ is full in $M$. Since $p$ is undersubscribed in both $M^\lor$ and $M'$, it follows that
$$
|M(p)| > |M^\lor(p)| \quad \text{and} \quad |M(p)| > |M'(p)|.
$$
Hence, there exists some student \( s' \in M(p) \setminus M^\lor(p) \), and in particular \( s' \in M(p) \setminus M'(p) \). Since \( s' \) is assigned in $M^\lor$ to the less preferred of their two projects from \( M \) and \( M' \), it follows that \( s' \) prefers \( M \) to \( M' \). Since both \( p \) and \( l \) are undersubscribed in \( M' \), the pair \( (s', p) \) blocks \( M' \), contradicting the stability of $M'$.
\end{itemize}

\noindent In both cases~(i) and (ii), a contradiction arises. Therefore, no blocking pair of type (S2 \& P1) exists in \( M^\lor \).\\

\noindent \textbf{(S2 \& P2):}  Here, \( s \) is assigned in \( M^\lor \), prefers \( p \) to \( M^\lor(s) \), \( p \) is undersubscribed in \( M^\lor \), \( l \) is full in \( M^\lor \), and \( s \in M^\lor(l) \). By Lemma~\ref{lemma_join_pjunder}, \( p \) is undersubscribed in at least one of \( M \) or \( M' \); without loss of generality, assume \( p \) is undersubscribed in \( M' \). Let \( p^* = M^\lor(s) \), where \( p^* \) is offered by \( l \). Since \( s \) is assigned in \( M^\lor \), they must be assigned in either \( M \) or \( M' \) (or both). We consider the possibilities for \( M(s) \) and \( M'(s) \):

\begin{itemize}
    \item[(i)] \( M'(s) = p^* \). Then \( s \) prefers \( p \) to \( M'(s) \). Moreover, it follows that \( s \in M'(l) \), and since \( p \) is undersubscribed in \( M' \), the pair \( (s, p) \) blocks \( M' \), a contradiction (This blocking pair occurs whether $l$ is undersubscribed or full in $M'$).

    \item[(ii)] \( M(s) = p^* \) and \( M'(s) \ne p^* \). Then \( s \in M(p^*) \setminus M'(p^*) \) and prefers \( p \) to \( M(s) \). If \( p \) is undersubscribed in \( M \), then since \( s \in M(l) \), the pair \( (s, p) \) blocks \( M \). Therefore, \( p \) must be full in \( M \). Following the case (ii) argument in (S2 \& P1) (where $p$ is full in $M$ but undersubscribed in both $M^\lor$ and $M'$), there exists some $s^* \in M(p) \setminus M'(p)$ who prefers $p$ to $M'(s^*)$. We now consider the following subcases:
    
    \item[(iia):] $s \in M'(l)$.  Recall that $s$ prefers $p$ to $M(s)$ and $p$ is full in $M$. By the stability of $M$, $l$ prefers the worst student in $M(p)$ to $s$; that is, $l$ prefers $s^*$ to $s$. Now since $s^*$ prefers $p$ to $M'(s)$, $p$ is undersubscribed in $M'$ and $l$ prefers $s^*$ to some student in $M'(l)$ (namely $s$), the pair $(s^*,p)$ blocks $M'$, a contradiction. (Again, this blocking pair occurs whether $l$ is undersubscribed or full in $M'$)

    \item[(iib:)] Suppose \(s \notin M'(l)\). Then \(s \in M(l) \setminus M'(l)\), so there exists some \(\hat{s} \in M'(l) \setminus M(l)\), since $|M(l)| = |M'(l)|$.  Since \(s \in M^\lor(l)\) and \(s \in M(l)\), it follows that \(s\) prefers \(M'\) to \(M\), i.e.\ \(s\) prefers \(M'(s)\) to \(p^*\).  Since \(p^*\) is undersubscribed in \(M'\) and \(\hat{s} \in M'(l) \setminus M(l)\), Lemma~\ref{lem:s-prefers-l-prefers1} implies that \(l\) prefers \(s\) to \(\hat{s}\). Now, recall that $s$ prefers $p$ to $M(s)$ and $p$ is full in $M$. By the stability of \(M\),  $l$ prefers the worst student in $M(p)$ to $s$; that is, $l$ prefers $s^*$ to $s$. Since, \(l\) prefers \(s^*\) to \(s\), it follows that \(l\) prefers \(s^*\) to \(\hat{s}\). However, \(s^*\) prefers \(p\) to \(M'(s^*)\),  \(p\) is undersubscribed in \(M'\), and \(l\) prefers \(s^*\) to some student in \(M'(l)\) (namely \(\hat{s}\)). Thus \((s^*,p)\) blocks \(M'\), a contradiction.  
    \end{itemize}

\smallskip
\noindent Therefore, no blocking pair of type (S2 \& P2) exists in \( M^\lor \).\\

\noindent \textbf{(S1 \& P3):}  
Suppose \( s \) is unassigned in \( M^\lor \), \( p \) is undersubscribed in \( M^\lor \), \( l \) is full in \( M^\lor \), and \( l \) prefers \( s \) to the worst student in \( M^\lor(l) \). By construction of \( M^\lor \), any student unassigned in \( M^\lor \) must also be unassigned in both \( M \) and \( M' \). Let \( s_z \) denote the worst student in \( M^\lor(l) \), so that \( l \) prefers \( s \) to \( s_z \). Since \( p \) is undersubscribed in \( M^\lor \), by Lemma~\ref{lemma_join_pjunder}, \( p \) is undersubscribed in at least one of \( M \) or \( M' \); without loss of generality, assume that \( p \) is undersubscribed in \( M' \) (it may be full or undersubscribed in \( M \)). Suppose first that \( s_z \in M'(l) \). Then in \( M' \), \( s \) is unassigned, \( p \) is undersubscribed, and  \( l \) prefers \( s \) to \( s_z \). Therefore, the pair \( (s, p) \) blocks \( M' \), contradicting its stability. It follows that \( s_z \notin M'(l) \), and hence \( s_z \in M(l) \setminus M'(l) \). 

\medskip
\noindent Since $s_z \in M^\lor(l)$, they must be assigned in $M^\lor$ to the less preferred of their two projects. Hence $s_z$ prefers $M'$ to $M$. Let $M(s_z) = p_z$.  If $p_z$ is full in $M'$, then $|M'(p_z)| \ge |M(p_z)|$. Since $s_z \in M(p_z) \setminus M'(p_z)$, there exists some $s^* \in M'(p_z) \setminus M(p_z)$. Since $s_z$ prefers $M'$ to $M$, Lemma~\ref{lem:s-prefers-l-prefers1} implies that $l$ prefers $s_z$ to $s^*$. Similarly, given that $s_z \in M(l) \setminus M'(l)$, there also exist some student $s^* \in M'(l) \setminus M(l)$.  Now if $p_z$ is undersubscribed in $M'$, then by the second part of Lemma~\ref{lem:s-prefers-l-prefers1}, $l$ prefers $s_z$ to $s^*$.  In both cases, $l$ prefers $s_z$ to $s^*$. Since $l$ also prefers $s$ to $s_z$, it follows that $l$ prefers $s$ to $s^*$. But in $M'$, $s$ is unassigned, $p$ is undersubscribed, and $l$ prefers $s$ to some $s^* \in M'(l)$. Hence $(s,p)$ blocks $M'$, a contradiction. Note that the pair $(s,p)$ blocks $M'$ whether $l$ is undersubscribed or full in $M'$.

\smallskip
\noindent It follows that no blocking pair of type (S1 \& P3) exists in \( M^\lor \).\\

\noindent \textbf{(S2 \& P3):}  Suppose \( s \) is assigned in \( M^\lor \), prefers \( p \) to \( M^\lor(s) \), \( p \) is undersubscribed in \( M^\lor \), \( l \) is full in \( M^\lor \), and \( l \) prefers \( s \) to the worst student in \( M^\lor(l) \). Let \( p^* = M^\lor(s) \), and suppose that \( p^* = M'(s) \), so \( s \) prefers \( p \) to \( p^* \). Let \( s_z \) be the worst student in \( M^\lor(l) \). By Lemma~\ref{lemma_join_pjunder}, \( p \) is undersubscribed in at least one of \( M \) or \( M' \). 

\medskip
\noindent Suppose first that \( p \) is undersubscribed in \( M' \) (it may be full or undersubscribed in \( M \)). If \( s_z \in M'(l) \), then in \( M' \),  \( s \) is assigned to \( p^* \), prefers \( p \) to \( p^* \), \( p \) is undersubscribed in $M'$, and \( l \) prefers \( s \) to \( s_z \). Thus, the pair \( (s, p) \) blocks \( M' \), a contradiction. It follows that \( s_z \in M(l) \setminus M'(l) \). Since \( s_z \in M^\lor(l) \), they are assigned in both \( M \) and \( M' \), and assigned in \( M^\lor \) to the less preferred of the two. Therefore, \( s_z \) prefers \( M' \) to \( M \). Let \( M(s_z) = p_z \), where \( p_z \) is offered by \( l \). Regardless of whether \( p_z \) is full or undersubscribed in \( M' \), by Lemma~\ref{lem:s-prefers-l-prefers1}, there exists some student \( s^* \in M'(l) \) such that \( l \) prefers \( s_z \) to \( s^* \). Since \( l \) prefers \( s \) to \( s_z \), it follows that \( l \) also prefers \( s \) to \( s^* \). However, since \( s \) prefers \( p \) to \( M'(s) \), and \( p \) is undersubscribed in \( M' \), the pair \( (s, p) \) blocks \( M' \), a contradiction.

\medskip
\noindent Now suppose instead that \( p \) is full in \( M' \) and undersubscribed in \( M \). Then \( |M'(p)| > |M^\lor(p)| \) and \( |M'(p)| > |M(p)| \), so there exists a student \( s^* \in M'(p) \setminus M^\lor(p) \), and in particular \( s^* \in M'(p) \setminus M(p) \). Moreover, $s^*$ prefers $p$ to $M(s)$. By the stability of $M$ and since $p$ is undersubscribed in $M$, it follows that $s^* \notin M(l)$, and $l$ prefers each student in $M(l)$ to $s^*$. Thus, \( s^* \in M'(l) \setminus M(l) \). If \( s_z \in M(l) \), then \( l \) prefers \( s_z \) to \( s^* \), and since \( l \) also prefers \( s \) to \( s_z \), it follows that \( l \) prefers \( s \) to \( s^* \). Hence, the pair \( (s, p) \) blocks \( M' \), a contradiction. We conclude that \( s_z \in M'(l) \setminus M(l) \). Since \( s_z \in M^\lor(l) \), it follows that \( s_z \) prefers \( M \) to \( M' \). Let \( M'(s_z) = p_z \), where \( p_z \) is offered by \( l \).

\medskip
\noindent Suppose \( p_z \) is full in \( M \). Then there exists a student \( \hat{s} \in M(p_z) \setminus M'(p_z) \). (This is because $s_z \in M'(p_z) \setminus M(p_z)$ and clearly $|M(p_z)| \ge |M'(p_z)|$). Since \( s_z \) prefers \( M \) to \( M' \), Lemma~\ref{lem:s-prefers-l-prefers1} implies that \( l \) prefers \( s_z \) to \( \hat{s} \). Moreover, since \( \hat{s} \in M(l) \), and \( l \) prefers each student in \( M(l) \) to \( s^* \), it follows that \( l \) prefers \( \hat{s} \) to \( s^* \), and \( l \) prefers \( s \) to \( s^* \) (since $l$ prefers $s$ to $s_z$ and $l$ prefers $s_z$ to $\hat{s}$). Suppose \( p_z \) is undersubscribed in \( M \). Then Lemma~\ref{lem:s-prefers-l-prefers1} implies that there exists a student \( \hat{s} \in M(l) \setminus M'(l) \) such that \( l \) prefers \( s_z \) to \( \hat{s} \). Since \( \hat{s} \in M(l) \), and $l$ prefers each student in $M(l)$ to $s^*$, it follows that $l$ prefers $\hat{s}$ to $s^*$, and consequently, prefers $s$ to $s^*$(since $l$ prefers $s$ to $s_z$ and $l$ prefers $s_z$ to $\hat{s}$). In both cases, \( l \) prefers \( s \) to some student \( s^* \in M'(p) \). Since \( s \) prefers \( p \) to \( M'(s) \), and \( p \) is full in \( M' \), the pair \( (s, p) \) blocks \( M' \), a contradiction.

\smallskip
\noindent A similar argument applies if \( p^* = M(s) \). We therefore conclude that no blocking pair of type (S2 \& P3) exists in \( M^\lor \). \\

\noindent \textbf{(S1 \& P4):}  Suppose \( s \) is unassigned in \( M^\lor \), project \( p \) is full in \( M^\lor \), and lecturer \( l \) prefers \( s \) to the worst student in \( M^\lor(p) \). Then, by construction of \( M^\lor \), student \( s \) is unassigned in both \( M \) and \( M' \). Let \( s_z \) denote the worst student in \( M^\lor(p) \), so that \( l \) prefers \( s \) to \( s_z \). Since \( s_z \in M^\lor(p) \), it must be that either \( (s_z, p) \in M \) or \( (s_z, p) \in M' \). 

\medskip
\noindent Suppose first that \( (s_z, p) \in M \). If \( p \) is full in \( M \), then \( s \) is unassigned, \( p \) is full, and \( l \) prefers \( s \) to \( s_z \in M(p) \), so the pair \( (s, p) \) blocks \( M \), contradicting its stability. If instead \( p \) is undersubscribed in \( M \), then \( s \) is unassigned, \( p \) is undersubscribed, and $l$ prefers $s$ to \( s_z \in M(l) \); thus, \( (s, p) \) again blocks \( M \), a contradiction. Now suppose that \( (s_z, p) \in M' \). The same reasoning applies: whether \( p \) is full or undersubscribed in \( M' \), student \( s \) is unassigned in \( M' \), \( s_z \in M'(p) \) and $s_z \in M'(l)$, and \( l \) prefers \( s \) to \( s_z \). Hence, the pair \( (s, p) \) blocks \( M' \), again a contradiction.

\smallskip
\noindent We conclude that no blocking pair of type (S1 \& P4) exists in \( M^\lor \).\\

\noindent \textbf{(S2 \& P4):}  Suppose \( s \) is assigned in \( M^\lor \), prefers \( p \) to \( M^\lor(s) \), \( p \) is full in \( M^\lor \), and \( l \) prefers \( s \) to the worst student in \( M^\lor(p) \). Let \( p^* = M^\lor(s) \), and suppose \( p^* = M'(s) \), so that \( s \) prefers \( p \) to \( p^* \). Let \( s_z \) be the worst student in \( M^\lor(p) \). If \( (s_z, p) \in M' \), then \( s \) prefers \( p \) to \( M'(s) \), \( p \) is full in \( M' \), and \( l \) prefers \( s \) to \( s_z \in M'(p) \). Moreover, if $p$ is undersubscribed in $M'$, then $l$ prefers $s$ to student $s_z \in M'(l)$. Therefore, the pair \( (s, p) \) blocks \( M' \), a contradiction. It follows that \( s_z \in M(p) \setminus M'(p) \). Since \( s_z \in M^\lor(p) \), they are assigned in both \( M \) and \( M' \), and assigned in \( M^\lor \) to the less preferred of the two. Hence, \( s_z \) prefers \( M' \) to \( M \).

\medskip
\noindent Suppose first that \( p \) is full in \( M' \). Then there exists some student \( s^* \in M'(p) \setminus M(p) \). Since \( s_z \) prefers \( M' \) to \( M \), Lemma~\ref{lem:s-prefers-l-prefers1} implies that \( l \) prefers \( s_z \) to \( s^* \). If instead \( p \) is undersubscribed in \( M' \), then by Lemma~\ref{lem:s-prefers-l-prefers1}, \( l \) prefers \( s_z \) to some student in \(s ^* \in M'(l) \setminus M(l) \). Since \( l \) also prefers \( s \) to \( s_z \), it follows that \( l \) prefers \( s \) to student \( s^* \in M'(l) \), namely \( s^* \). Therefore, regardless of whether \( p \) is full or undersubscribed in \( M' \), the pair \( (s, p) \) blocks \( M' \), a contradiction. A similar argument applies if \( p^* = M(s) \). We therefore conclude that no blocking pair of type (S2 \& P4) can exist in \( M^\lor \).

\medskip
\noindent In all possible cases, we derive a contradiction. Therefore, no blocking pair exists in \( M^\lor \), and \( M^\lor \) is stable.
\end{proof}

\medskip
\noindent We denote by $M \lor M'$ the set of (student, project) pairs in which each student is assigned to the poorer of her projects in $M$ and $M'$; and it follows from Lemma $\ref{lemma_join_stable}$ that if each student is given the poorer of her projects in any fixed set of stable matchings, then the resulting assignment is a stable matching. For the case where $\mathcal{M}$ is the set of all stable matchings in $I$, we denote by $\bigvee_{M \in \mathcal{M}}M$, or simply $\bigvee \mathcal{M}$, the resulting stable matching.  This matching is student-pessimal and, by Corollary~\ref{cor:dominance-reversal}, lecturer-optimal. We are now ready to present our main result.

\begin{theorem}
\label{theorem-lattice}
    Let $I$ be an instance of {\sc spa-s}, and let $\mathcal{M}$ be the set of stable matchings in $I$. Let $\preceq$ be the dominance partial order on $\mathcal{M}$ and let $M, M' \in \mathcal{M}$. Then $(\mathcal{M},\preceq)$ is a distributive lattice, with $M \land M'$ representing the meet of $M$ and $M'$, and $M \lor M'$ the join of $M$ and $M'$.
\end{theorem}

\begin{proof}
We note that the proof is analogous to that in \cite{O2020}. Let $M$ and $M'$ be two stable matchings in $\mathcal{M}$. By Lemma $\ref{lemma_meet_stable}$, we have that $M \land M'$ is a stable matching; and by the definition of $M \land M'$, it follows that $M \land M' \preceq M$ and $M \land M' \preceq M'$. Further, if $M^*$ is an arbitrary stable matching satisfying $M^* \preceq M$ and $M^* \preceq M'$, then each student must be assigned in $M^*$ to a project that is at least as good as her assigned projects in each of $M$ and $M'$, so that $M^* \preceq M \land M'$. Thus $M \land M'$ is the meet of $M$ and $M'$. Similarly, by Lemma $\ref{lemma_join_stable}$, we have that $M \lor M'$ is a stable matching; and by the definition of $M \lor M'$, it follows that $M \preceq M \lor M'$ and $M' \preceq M \lor M'$. Following a similar argument as above, $M \lor M'$ is the join of $M$ and $M'$. Hence, $(\mathcal{M}, \preceq)$ is a lattice.\\

    \noindent Next, we show that the join and meet operation distribute over each other. Let $M,M'$ and $M''$ be stable matchings in $\mathcal{M}$. First, let $X = M \lor (M' \land M'')$ and let $Y = (M \lor M') \land (M \lor M'')$; we need to show that $X = Y$. Let $s_i$ be an arbitrary student. If $s_i$ is unassigned in each of $M$, $M'$ and $M''$, it is clear that $s_i$ is unassigned in both $X$ and $Y$. Now, suppose that $s_i$ is assigned to some project in each of $M$, $M'$ and $M''$. We consider the following cases.
    \begin{enumerate}[label = (\roman*)]
        \item Suppose that $M(s_i) = M'(s_i) = M''(s_i)$, clearly $X(s_i) = Y(s_i)$.
        \item Suppose that either (a) $M(s_i) = M'(s_i)$ and $M(s_i) \neq M''(s_i)$ or (b) $M(s_i) \neq M'(s_i)$ and $M(s_i) = M''(s_i)$ holds. Irrespective of how we express $s_i$'s preference over $M(s_i),  M'(s_i)$ and $M''(s_i)$ in cases (a) and (b), we have that $s_i$ is assigned to $M(s_i)$ in both $X$ and $Y$.
        \item Suppose that $M'(s_i) = M''(s_i)$ and $M'(s_i) \neq M(s_i)$. If $s_i$ prefers $M'(s_i)$ to $M(s_i)$ then $s_i$ is assigned to $M(s_i)$ in both $X$ and $Y$. Otherwise, if $s_i$ prefers $M(s_i)$ to $M'(s_i)$ then $s_i$ is assigned to $M'(s_i)$ in both $X$ and $Y$.
        \item Suppose that $M(s_i)$, $M'(s_i)$ and $M''(s_i)$ are distinct projects. There are six different ways to express $s_i$'s preference over $M(s_i)$, $M'(s_i)$ and $M''(s_i)$. If $s_i$ prefers $M(s_i)$ to $M'(s_i)$ to $M''(s_i)$, then $s_i$ is assigned to $M'(s_i)$ in both $X$ and $Y$. If $s_i$ prefers $M(s_i)$ to $M''(s_i)$ to $M'(s_i)$, then $s_i$ is assigned to $M''(s_i)$ in both $X$ and $Y$. We leave it to the reader to verify that in the remaining four cases, $s_i$ is assigned to $M(s_i)$ in both $X$ and $Y$.
    \end{enumerate}
\noindent Since $s_i$ is an arbitrary student, it follows that $X = Y$; and thus the first distributive property holds. Next, we show that the second distributive property holds. Let $X = M \land (M' \lor M'')$ and let $Y = (M \land M') \lor (M \land M'')$. Let $s_i$ be an arbitrary student. Again, if $s_i$ is unassigned in each of $M,M'$ and $M''$, it is clear that $s_i$ is unassigned in both $X$ and $Y$. Now, suppose $s_i$ is assigned to some project in each of $M,M'$ and $M''$. Following the same case analysis as before, we arrive at the same conclusion in cases (i) and (ii). We consider cases (iii) and (iv) in detail:

\begin{enumerate}[label = (\roman*)]
\setcounter{enumi}{2}
    \item If $M'(s_i) = M''(s_i)$ and $M'(s_i) \neq M(s_i)$. If $s_i$ prefers $M'(s_i)$ to $M(s_i)$ then $s_i$ is assigned to $M'(s_i)$ in both $X$ and $Y$. Otherwise, if $s_i$ prefers $M(s_i)$ to $M'(s_i)$ then $s_i$ is assigned to $M(s_i)$ in both $X$ and $Y$.
    \item Suppose that $M(s_i)$, $M'(s_i)$ and $M''(s_i)$ are distinct projects. Again, there are six different ways to express $s_i$'s preference over $M(s_i)$, $M'(s_i)$ and $M''(s_i)$. If $s_i$ prefers $M'(s_i)$ to $M''(s_i)$ to $M(s_i)$, then $s_i$ is assigned to $M''(s_i)$ in both $X$ and $Y$. If $s_i$ prefers $M''(s_i)$ to $M'(s_i)$ to $M(s_i)$, then $s_i$ is assigned to $M'(s_i)$ in both $X$ and $Y$. We leave it to the reader to verify that in the remaining four cases, $s_i$ is assigned to $M(s_i)$ in both $X$ and $Y$.
\end{enumerate}
\noindent Since $s_i$ is an arbitrary student, it follows that $X = Y$; and thus the second distributive property holds. Since each of $M$, $M'$ and $M''$ is an arbitrary stable matching in $\mathcal{M}$, it follows that $(\mathcal{M},\preceq)$ is a distributive lattice.
\end{proof}

\subsection{Example}
\noindent Finally, consider the {\sc spa-s} instance \( I_3 \) illustrated in Figure \ref{fig:example3}, which admits a total of seven stable matchings, as shown in Table \ref{tab:table3}. The \textit{meet} of matchings \( M_3 \) and \( M_4 \) is \( M_2 \), i.e., \( M_2 = M_3 \land M_4 \). For each student assigned to different projects in \( M_3 \) and \( M_4 \)—namely, \( s_2 \), \( s_4 \), \( s_5 \), \( s_6 \), and \( s_7 \)—the assignment in \( M_2 \) corresponds to the better of their projects in \( M_3 \) and \( M_4 \). Conversely, the \textit{join} of matchings \( M_3 \) and \( M_4 \) is \( M_5 \), i.e., \( M_5 = M_3 \lor M_4 \). In \( M_5 \), each student is assigned to the poorer of their projects in \( M_3 \) and \( M_4 \).

\begin{figure}[H]
\centering
\renewcommand{\arraystretch}{1.2} 
\setlength{\tabcolsep}{5pt} 
\resizebox{\textwidth}{!}{ 
\begin{tabular}{p{0.33\textwidth} p{0.41\textwidth} p{0.25\textwidth}}
\hline
\textbf{Students' preferences} & \textbf{Lecturers' preferences} & \textbf{Offers} \\ 
\hline
$s_1$: $p_1 \; p_2 \; p_4 \; p_3$ & $l_1$: $s_7 \; s_9 \; s_3 \; s_4 \; s_5 \; s_1 \; s_2 \; s_6 \; s_8 $ & $p_1$, $p_2$, $p_5$, $p_6$ \\ 

$s_2$: $p_1 \; p_4 \; p_3 \; p_2$ & $l_2$: $s_6 \; s_1 \; s_2 \; s_5 \; s_3 \; s_4 \; s_7 \; s_8 \;  s_9$ & $p_3$, $p_4$, $p_7$, $p_8$ \\ 

$s_3$: $p_3 \; p_1 \; p_2 \; p_4$ & & \\ 
$s_4$: $p_3 \; p_2 \; p_1 \; p_4$ & & \\ 
$s_5$: $p_4 \; p_3 \; p_1$ & & \\ 
$s_6$: $p_5 \; p_2 \; p_7$ & & \\ 
$s_7$: $p_7 \; p_3 \; p_6$ & & \\ 
$s_8$: $p_6 \; p_8$ & \multicolumn{2}{l}{\textbf{Project capacities:} $c_1 = c_3 = 2$; $\forall j \in \{2, 4, 5, 6, 7, 8\}, \, c_j = 1$} \\ 
$s_9$: $p_8 \; p_2 \; p_3$ & \multicolumn{2}{l}{\textbf{Lecturer capacities:} $d_1 = 4$, $d_2 = 5$}  \\ 
\hline
\end{tabular}
}
\caption{An instance \( I_3 \) of {\sc spa-s}} 
\label{fig:example3}
\end{figure}


\begin{table}[H]
    \centering
    \setlength{\tabcolsep}{6pt} 
    \renewcommand{\arraystretch}{1.3} 
    \begin{tabular}{c|*{9}{c}} 
        Matching & $s_1$ & $s_2$ & $s_3$ & $s_4$ & $s_5$ & $s_6$ & $s_7$ & $s_8$ & $s_9$ \\
        \hline
        $M_1$ & $p_1$ & $p_1$ & $p_3$ & $p_3$ & $p_4$ & $p_5$ & $p_7$ & $p_6$ & $p_8$ \\
        $M_2$ & $p_1$ & $p_1$ & $p_3$ & $p_3$ & $p_4$ & $p_5$ & $p_7$ & $p_8$ & $p_2$  \\
        $M_3$ & $p_1$ & $p_1$ & $p_3$ & $p_3$ & $p_4$ & $p_7$ & $p_6$ & $p_8$ & $p_2$ \\
        $M_4$ & $p_1$ & $p_4$ & $p_3$ & $p_1$ & $p_3$ & $p_5$ & $p_7$ & $p_8$ & $p_2$ \\
        $M_5$ & $p_1$ & $p_4$ & $p_3$ & $p_1$ & $p_3$ & $p_7$ & $p_6$ & $p_8$ & $p_2$\\
        $M_6$ & $p_4$ & $p_3$ & $p_1$ & $p_1$ & $p_3$ & $p_5$ & $p_7$ & $p_8$ & $p_2$  \\
        $M_7$ & $p_4$ & $p_3$ & $p_1$ & $p_1$ & $p_3$ & $p_7$ & $p_6$ & $p_8$ & $p_2$ \\
    \end{tabular}
    \caption{Instance $I_3$ admits seven stable matchings.}
    \label{tab:table3}
\end{table}

\vspace{1cm}
\noindent The Hasse diagram illustrated in Figure \ref{fig:instance-2-lattice} is a directed graph with each vertex representing a stable matching, and there is a directed edge from vertex $M$ to $M'$ if $M \preceq M'$ and there is no such $M^*$ such that $M \preceq M^* \preceq M'$. We note that all the edges representing precedence implied by
transitivity are suppressed in the diagram.\\

\begin{figure}[!ht]
    \centering
\begin{tikzpicture}[scale=0.80]
\SetVertexNormal[MinSize = 1pt,LineWidth = 1pt,FillColor = white]
\Vertex[x=5, y=7]{$M_1$}
\Vertex[x=5, y=5]{$M_2$}
\Vertex[x=3, y=3]{$M_3$}
\Vertex[x=7, y=3]{$M_4$}
\Vertex[x=2, y=1]{$M_5$}
\Vertex[x=8, y=1]{$M_6$}
\Vertex[x=5, y=-1]{$M_7$}

\draw [thick, ->, > = stealth, shorten <= 0.4cm, shorten >= 0.4cm]  (5,7) to (5,5); 
\draw [thick, ->, > = stealth, shorten <= 0.4cm, shorten >= 0.4cm]  (5,5) to (3,3); 
\draw [thick, ->, > = stealth, shorten <= 0.4cm, shorten >= 0.4cm]  (5,5) to (7,3); 
\draw [thick, ->, > = stealth, shorten <= 0.4cm, shorten >= 0.4cm]  (3,3) to (2,1); 
\draw [thick, ->, > = stealth, shorten <= 0.4cm, shorten >= 0.4cm]  (7,3) to (2,1); 
\draw [thick, ->, > = stealth, shorten <= 0.4cm, shorten >= 0.4cm]  (7,3) to (8,1); 
\draw [thick, ->, > = stealth, shorten <= 0.4cm, shorten >= 0.4cm]  (2,1) to (5,-1); 
\draw [thick, ->, > = stealth, shorten <= 0.4cm, shorten >= 0.4cm]  (8,1) to (5,-1); 
\end{tikzpicture}
\caption{Lattice structure for $I_3$}
\label{fig:instance-2-lattice}
\end{figure}

\section{Conclusion and Open Problems}
\label{sect:conclusion}
\noindent In this paper, we study the structure of the set of stable matchings in a given instance $I$ of {\sc spa-s}. We proved that the set of stable matchings, ordered by the defined dominance relation, forms a finite distributive lattice. This result reveals that stable matchings in {\sc spa-s} exhibit a well-defined structure with desirable combinatorial properties. While such a lattice structure is well known in classical bipartite models like the Stable Marriage problem, our contribution is novel in extending this property to a more complex setting involving three types of agents: students, projects, and lecturers. Additionally, we presented several structural properties of {\sc spa-s} instances that, to the best of our knowledge, have not been previously studied.

\bigskip
\noindent Our results on the lattice structure of stable matchings in {\sc spa-s} also enable a formal definition of meta-rotations, and support the construction of the meta-rotation poset—a compact representation of the complete set of stable matchings in any {\sc spa-s} instance. Together with future work on meta-rotations, the results presented here provide a foundation for the development of efficient algorithms to enumerate all stable matchings, identify all stable pairs, and compute egalitarian stable matchings in {\sc spa-s}, similar to results in the Stable Marriage and Hospital Residents models. An open question remains as to whether an analogous lattice structure exists in {\sc spa-st}, the extension of {\sc spa-s} that allows ties in student and lecturer preferences, under super and strong stability \cite{O2020}. The results in this paper offer a starting point for addressing this question and for exploring the complexity of related problems in {\sc spa-s} and its extensions.

\section{Acknowledgements}
\noindent The authors thank the anonymous ISCO reviewers for their valuable feedback, which contributed to improving the clarity and presentation of the conference version of this paper. We would also like to appreciate Betina Klaus for her contributions to an earlier version of this work.

\bibliographystyle{elsarticle-num-names} 
\bibliography{ref}

\end{document}